\documentclass[submission,copyright,creativecommons]{eptcs}
\providecommand{\event}{QPL2024}

\input{preamble}
\newcommand{\md}[1]{\ensuremath{\ (\text{mod } #1)}}

\usepackage{iftex}

\ifpdf
  \usepackage{underscore}         
  \usepackage[T1]{fontenc}        
  \usepackage[utf8]{inputenc}
\else
  \usepackage{breakurl}           
\fi

\title{Scalable Spider Nests (...Or How to Graphically Grok Transversal Non-Clifford Gates)}

\author{Aleks Kissinger
\institute{University of Oxford}
\email{aleks.kissinger@ox.ac.uk} \and
John van de Wetering
\institute{University of Amsterdam}
\email{john@vdwetering.name}
}

\def\titlerunning{Scalable spider nests}
\def\authorrunning{A.~Kissinger \& J.~van de Wetering}

\begin{document}
\maketitle

\begin{abstract}
This is the second in a series of ``graphical grokking'' papers in which we study how stabiliser codes can be understood using the ZX-calculus. In this paper we show that certain complex rules involving ZX-diagrams, called spider nest identities, can be captured succinctly using the \emph{scalable} ZX-calculus, and all such identities can be proved inductively from a single new rule using the Clifford ZX-calculus. This can be combined with the ZX picture of CSS codes, developed in the first ``grokking'' paper, to give a simple characterisation of the set of all transversal diagonal gates at the third level of the Clifford hierarchy implementable in an arbitrary CSS code.
\end{abstract}

\noindent The ZX-calculus~\cite{CD1,CD2} is a useful tool for expressing and reasoning about quantum computations. It represents computations, such as quantum circuits or measurement patterns, as certain labelled open graphs called \textit{ZX-diagrams}, subject to a collection of rewrite rules that can be used to transform and simplify diagrams. While the literature tends to talk about ``the'' ZX-calculus, there are actually several calculi of increasing power. One of the simpler versions is what we will call here the \textit{Clifford ZX-calculus}, which is complete for Clifford ZX-diagrams, i.e. diagrams whose phase parameters are restricted to integer multiples of $\pi/2$~\cite{BackensCompleteness}. These give a natural generalisation of Clifford circuits, and the Clifford ZX-calculus admits many efficient algorithms, e.g. for reducing Clifford diagrams to normal form, equality checking, and computing arbitrary measurement amplitudes. In some sense, it is the graphical counterpart to stabiliser theory (see e.g.~\cite{BackensCompleteness,borghans2019zx}), so it gives a natural setting for studying quantum error correction~\cite{kissinger2022grok}.

However, the Clifford ZX-calculus, unlike some of its more powerful counterparts~\cite{HarnyCompleteness,jeandel2019completeness} is known to be incomplete for ZX-diagrams whose phases are not all multiples of $\pi/2$. An interesting family of equations that are not provable in the Clifford ZX-calculus are the \textit{spider-nest identities}, where elaborate configurations of \textit{phase gadgets} can collectively cancel out with each other~\cite{debeaudrap2020spidernest2}. For example, connecting a $\pi/4$ phase gadget to all $2^4-1 = 15$ non-empty subsets of 4 qubits is equal to the identity:
\begin{equation}\label{eq:sn-1-4}
  \begin{tikzpicture}
	\begin{pgfonlayer}{nodelayer}
		\node [style=Z phase dot] (0) at (-3.5, 2) {$\frac\pi 4$};
		\node [style=none] (4) at (-8, 2) {};
		\node [style=none] (5) at (1, 2) {};
		\node [style=none] (7) at (-8, 0.5) {};
		\node [style=none] (8) at (1, 0.5) {};
		\node [style=none] (10) at (-8, -1) {};
		\node [style=none] (11) at (1, -1) {};
		\node [style=none] (13) at (-8, -2.5) {};
		\node [style=none] (14) at (1, -2.5) {};
		\node [style=Z phase dot] (15) at (-3.5, 0.5) {$\frac\pi 4$};
		\node [style=Z phase dot] (16) at (-3.5, -1) {$\frac\pi 4$};
		\node [style=Z phase dot] (17) at (-3.5, -2.5) {$\frac\pi 4$};
		\node [style=X dot] (28) at (-3.5, 3.25) {};
		\node [style=Z phase dot] (29) at (-3.5, 4) {$\frac\pi 4$};
		\node [style=X dot] (30) at (-1.25, 1.25) {};
		\node [style=Z phase dot] (31) at (-0.25, 1.25) {$\frac\pi 4$};
		\node [style=X dot] (32) at (-1.25, -3.25) {};
		\node [style=Z phase dot] (33) at (-0.25, -3.25) {$\frac\pi 4$};
		\node [style=X dot] (34) at (-1.25, -1.75) {};
		\node [style=Z phase dot] (35) at (-0.25, -1.75) {$\frac\pi 4$};
		\node [style=X dot] (36) at (-1.25, -0.25) {};
		\node [style=Z phase dot] (37) at (-0.25, -0.25) {$\frac\pi 4$};
		\node [style=X dot] (38) at (-5.75, -3.25) {};
		\node [style=Z phase dot] (39) at (-6.75, -3.25) {$\frac\pi 4$};
		\node [style=X dot] (40) at (-5.75, 2.75) {};
		\node [style=Z phase dot] (41) at (-6.75, 2.75) {$\frac\pi 4$};
		\node [style=X dot] (42) at (-1.25, 2.75) {};
		\node [style=Z phase dot] (43) at (-0.25, 2.75) {$\frac\pi 4$};
		\node [style=X dot] (44) at (-5.75, -1.75) {};
		\node [style=Z phase dot] (45) at (-6.75, -1.75) {$\frac\pi 4$};
		\node [style=X dot] (46) at (-5.75, 1.25) {};
		\node [style=Z phase dot] (47) at (-6.75, 1.25) {$\frac\pi 4$};
		\node [style=X dot] (48) at (-5.75, -0.25) {};
		\node [style=Z phase dot] (49) at (-6.75, -0.25) {$\frac\pi 4$};
		\node [style=none] (50) at (3, 2) {};
		\node [style=none] (51) at (6.25, 2) {};
		\node [style=none] (52) at (3, 0.5) {};
		\node [style=none] (53) at (6.25, 0.5) {};
		\node [style=none] (54) at (3, -1) {};
		\node [style=none] (55) at (6.25, -1) {};
		\node [style=none] (56) at (3, -2.5) {};
		\node [style=none] (57) at (6.25, -2.5) {};
		\node [style=none] (58) at (2, -0.25) {$=$};
	\end{pgfonlayer}
	\begin{pgfonlayer}{edgelayer}
		\draw (4.center) to (5.center);
		\draw (7.center) to (8.center);
		\draw (10.center) to (11.center);
		\draw (13.center) to (14.center);
		\draw (28) to (29);
		\draw (28) to (0);
		\draw [bend right] (28) to (17);
		\draw [bend right] (28) to (15);
		\draw [bend left] (16) to (28);
		\draw (30) to (31);
		\draw (32) to (33);
		\draw (34) to (35);
		\draw (36) to (37);
		\draw (0) to (30);
		\draw (15) to (30);
		\draw (16) to (30);
		\draw (15) to (32);
		\draw (16) to (32);
		\draw (17) to (32);
		\draw (0) to (34);
		\draw (16) to (34);
		\draw (17) to (34);
		\draw (0) to (36);
		\draw (15) to (36);
		\draw (17) to (36);
		\draw (38) to (15);
		\draw (38) to (17);
		\draw (39) to (38);
		\draw (41) to (40);
		\draw (40) to (0);
		\draw (40) to (16);
		\draw (43) to (42);
		\draw (42) to (0);
		\draw (42) to (17);
		\draw (17) to (44);
		\draw (44) to (16);
		\draw (45) to (44);
		\draw (47) to (46);
		\draw (46) to (0);
		\draw (46) to (15);
		\draw (49) to (48);
		\draw (48) to (15);
		\draw (48) to (16);
		\draw (50.center) to (51.center);
		\draw (52.center) to (53.center);
		\draw (54.center) to (55.center);
		\draw (56.center) to (57.center);
	\end{pgfonlayer}
\end{tikzpicture}
\end{equation}
While this might seem like a very specific and complicated rule, such identities have already been used to great effect for T-gate optimisation, originally in the language of phase polynomials~\cite{heyfron2018efficient,amy2016t} and later explicitly as ZX-diagram rules~\cite{debeaudrap2020spidernest2}.

The starting point for this paper is the observation that all spider-nest identities can be succinctly characterised using certain boolean matrices called \textit{triorthogonal} matrices, i.e.~those where the Hamming weight of the product of any $\ell\leq 3$ columns is zero modulo $2^\ell$. Triorthogonal matrices have been used extensively in the study of transversal gates for quantum error-correcting codes and magic state distillation~\cite{bravyi2005universal,bravyi2012magic,nezami2022classification}. This connection is not really new, and could probably best be described as folklore, as it really just explicitly connects the dots between two important results in the literature: the equivalence between T-count optimisation and Reed-Muller decoding~\cite{amy2016t} and the equivalence between triorthogonal matrices and certain Reed-Muller codewords~\cite{nezami2022classification}.

We show that the equivalence between spider nest identities and triorthogonal matrices can be made explicit and fully graphical, with the help of the \textit{scalable ZX-calculus}~\cite{SZXCalculus}. This extension to ZX notation enables one to work with entire registers of qubits at once to represent arbitrary connectivity between components using boolean biadjacency matrices. Notably, all spider nest equations take a simple form, for some triorthogonal matrix $M$:
\begin{equation}\label{eq:triortho-map}
\begin{tikzpicture}
	\begin{pgfonlayer}{nodelayer}
		\node [style=none] (0) at (-3, 0) {};
		\node [style=Z bold dot] (1) at (-2, 0) {};
		\node [style=none] (2) at (-1, 0) {};
		\node [style=Z bold phase dot] (3) at (-2, 1.875) {$\frac\pi4$};
		\node [style=umat] (4) at (-2, 0.75) {};
		\node [style=none] (5) at (-1.25, 0.75) {\footnotesize $M$};
		\node [style=none] (10) at (0, 0) {$=$};
		\node [style=none] (11) at (1, 0) {};
		\node [style=none] (12) at (2.5, 0) {};
	\end{pgfonlayer}
	\begin{pgfonlayer}{edgelayer}
		\draw [style=boldedge] (0.center) to (2.center);
		\draw [style=boldedge] (1) to (3);
		\draw [style=boldedge] (11.center) to (12.center);
	\end{pgfonlayer}
\end{tikzpicture}
\end{equation}

We show that we only have to assume one further rule, a  variant of Eq.~\eqref{eq:sn-1-4}, beyond the standard Clifford ZX-calculus in order to prove all identities of the form of Eq.~\eqref{eq:triortho-map} (cf.~Theorem~\ref{thm:diag-complete}). This in turn gives us a complete calculus for CNOT+$T$ circuits (cf.~Corollary~\ref{cor:CNOT+T}).

The graphical form for spider-nest identities~\eqref{eq:triortho-map} is particularly handy when used in the context of error correcting codes. As shown in the first ``graphical grokking'' paper~\cite{kissinger2022grok}, the encoding map of a CSS code can be represented as a ZX-diagram in a certain normal form. By ``pushing'' maps through the encoder, one can compute the effect of a physical map on the logical qubits or vice-versa. 
Of particular importance for fault-tolerant computation are those logical maps that can be implemented by a \textit{transversal} operation on the physical qubits, which in the setting we will consider are operations implemented by a tensor product of single-qubit unitaries. Codes with rich sets of transversal logical gates are interesting both for supporting computation on their own encoded qubits and as part of \textit{magic state distillation} protocols, used to boost (non-universal) fault-tolerant computation in other codes to universality~\cite{Litinski2019gameofsurfacecodes}.

Many characterisations of transversal gates exist in the literature for specific classes of gates and/or codes~\cite{bravyi2005universal,bravyi2012magic,campbell2017unified,rengaswamy2020optimality,vuillot2022quantum}. The one we give here is essentially equivalent to the one given in~\cite{webster2023transversal} concerning diagonal gates in the Clifford hierarchy, although we will focus on just the third level of the Clifford hierarchy for the sake of simplicity. Namely, for a fixed CSS code, we give a complete classification for the set of transversal gates whose logical action is in $\mathcal D_3$, the diagonal unitaries on the third level of the Clifford hierarchy (Theorem~\ref{thm:char-transversal}). While we focus on $\mathcal D_3$, the method we show translates straightforwardly to $\mathcal D_{\ell}$, the diagonal unitaries of the $\ell$th level of the Clifford hierarchy, and phase gadgets with a $\pi/2^{\ell-1}$ phase, as we will remark in the conclusion. A notable feature of our characterisation is not so much the result itself, but the proof technique, which demonstrates the interplay between the CSS code and the triorthogonal structure that needs to be present in it. Both of these can be treated uniformly using scalable notation, and the graphical rules allow one to easily see (and hopefully grok) how the stabiliser and non-stabiliser aspects of the computation interact.

\section{Preliminaries}

The basic building blocks of ZX-diagrams are spiders, which come in two varieties, \textit{Z spiders} and \textit{X spiders}, defined respectively relative to the eigenbases of the Pauli Z and Pauli X operators.
\begin{equation}\label{eq:spiders}
  \begin{array}{ccccl}
    \zsp{m}{n}{\alpha}
    & \ := \ &
    \textrm{\scriptsize $m$}\!\left\{\ \ \begin{tikzpicture}
	\begin{pgfonlayer}{nodelayer}
		\node [style=none] (1) at (-1.5, 1) {};
		\node [style=none] (2) at (-1.5, -1) {};
		\node [style=none, rotate=90] (6) at (-1.25, 0) {...};
		\node [style=Z phase dot] (7) at (0, 0) {$\alpha$};
		\node [style=none] (8) at (1.5, 1) {};
		\node [style=none] (9) at (1.5, -1) {};
		\node [style=none, rotate=90] (11) at (1.25, 0) {...};
	\end{pgfonlayer}
	\begin{pgfonlayer}{edgelayer}
		\draw [in=180, out=45] (7) to (8.center);
		\draw [in=180, out=-45] (7) to (9.center);
		\draw [in=135, out=0] (1.center) to (7);
		\draw [in=225, out=0] (2.center) to (7);
	\end{pgfonlayer}
\end{tikzpicture}\ \ \right\}\!\textrm{\scriptsize $n$}
    & \ = \ &
    \ket{0}^{\otimes n}\bra{0}^{\otimes m} + e^{i\alpha}\ket{1}^{\otimes n}\bra{1}^{\otimes m} \\[5mm]
    \xsp{m}{n}{\alpha}
    & \ := \ &
    \textrm{\scriptsize $m$}\!\left\{\ \ \begin{tikzpicture}
	\begin{pgfonlayer}{nodelayer}
		\node [style=none] (1) at (-1.5, 1) {};
		\node [style=none] (2) at (-1.5, -1) {};
		\node [style=none, rotate=90] (6) at (-1.25, 0) {...};
		\node [style=X phase dot] (7) at (0, 0) {$\alpha$};
		\node [style=none] (8) at (1.5, 1) {};
		\node [style=none] (9) at (1.5, -1) {};
		\node [style=none, rotate=90] (11) at (1.25, 0) {...};
	\end{pgfonlayer}
	\begin{pgfonlayer}{edgelayer}
		\draw [in=180, out=45] (7) to (8.center);
		\draw [in=180, out=-45] (7) to (9.center);
		\draw [in=135, out=0] (1.center) to (7);
		\draw [in=225, out=0] (2.center) to (7);
	\end{pgfonlayer}
\end{tikzpicture}\ \ \right\}\!\textrm{\scriptsize $n$}
    & \ = \ &
    \ket{{+}}^{\otimes n}\bra{{+}}^{\otimes m} +
    e^{i\alpha}\ket{{-}}^{\otimes n}\bra{{-}}^{\otimes m}
  \end{array}
\end{equation}
where $\<\psi|^{\otimes m}$ and $|\psi\>^{\otimes n}$ are the $m$- and $n$-fold tensor products of bras and kets, respectively, and we take the convention that $(...)^{\otimes 0} = 1$. The parameter $\alpha$ is called the \textit{phase} of a spider. If we omit the phase, it is assumed to be $0$.
In addition to spiders, we allow identity wires, swaps, cups, and caps in ZX-diagrams, which are defined as follows:
\[
  \tikz[tikzfig]{ \draw (0,0) -- (3,0); } \ \ :=\ \  \sum_i |i\>\<i|\qquad
  \begin{tikzpicture}
	\begin{pgfonlayer}{nodelayer}
		\node [style=none] (6) at (-1.5, 0.75) {};
		\node [style=none] (7) at (-1.5, -0.75) {};
		\node [style=none] (12) at (1.5, -0.75) {};
		\node [style=none] (14) at (1.5, 0.75) {};
	\end{pgfonlayer}
	\begin{pgfonlayer}{edgelayer}
		\draw [in=180, out=0, looseness=0.75] (6.center) to (12.center);
		\draw [in=-180, out=0] (7.center) to (14.center);
	\end{pgfonlayer}
\end{tikzpicture} \ \ :=\ \ \sum_{ij} |ij\>\<ji| \qquad
  \begin{tikzpicture}
	\begin{pgfonlayer}{nodelayer}
		\node [style=none] (12) at (0, -0.75) {};
		\node [style=none] (14) at (0, 0.75) {};
		\node [style=none] (15) at (0, 0.75) {};
		\node [style=none] (16) at (0, -0.75) {};
	\end{pgfonlayer}
	\begin{pgfonlayer}{edgelayer}
		\draw [bend left=90, looseness=1.75] (16.center) to (15.center);
	\end{pgfonlayer}
\end{tikzpicture}
 \ \ := \ \ \sum_i |ii\> \qquad
  \begin{tikzpicture}
	\begin{pgfonlayer}{nodelayer}
		\node [style=none] (12) at (0, -0.75) {};
		\node [style=none] (14) at (0, 0.75) {};
		\node [style=none] (15) at (0, 0.75) {};
		\node [style=none] (16) at (0, -0.75) {};
	\end{pgfonlayer}
	\begin{pgfonlayer}{edgelayer}
		\draw [bend right=90, looseness=1.75] (16.center) to (15.center);
	\end{pgfonlayer}
\end{tikzpicture}
 \ \ := \ \ \sum_i \<ii|
\]
If all of the angles in a ZX-diagram are integer multiples of $\pi/2$, it is called a \textit{Clifford ZX-diagram}. If not, it is called a \textit{non-Clifford ZX-diagram}. There is a direct translation from Clifford+phase circuits to ZX-diagrams, where the resulting diagram is Clifford if and only if the circuit contains no non-Clifford phase gates.

ZX-diagrams have the useful property that they are invariant under arbitrary deformations and swapping input/output wires of spiders. This property is sometimes referred to as \textit{only connectivity matters}. In addition to this ``meta-rule'', the \textit{Clifford ZX-calculus} consists of the 7 rules shown in Figure~\ref{fig:zx-rules}. Notably, this set of rules is \textit{complete} for Clifford ZX-diagrams. That is, if two Clifford ZX-diagrams describe the same linear map, one can transform one into the other using the Clifford ZX-calculus. For the Clifford ZX-calculus, this transformation is furthermore efficient. Note that we presented here the rules only up to non-zero scalar factor (denoted by $\scalareq$). Scalars will not play an important role in this paper.

In addition to encodings of basic gates, a useful unitary ZX-diagram is a \textit{phase gadget}, which applies a relative phase of $\alpha$ to all of the computational basis states whose bitstring has parity~1~\cite{kissinger2019tcount}:
\[
  \begin{tikzpicture}
	\begin{pgfonlayer}{nodelayer}
		\node [style=none] (0) at (-1.25, 1) {};
		\node [style=none] (1) at (1.5, 1) {};
		\node [style=none] (2) at (-1.25, -1) {};
		\node [style=none] (3) at (1.5, -1) {};
		\node [style=Z dot] (4) at (0, 1) {};
		\node [style=Z dot] (5) at (0, -1) {};
		\node [style=X dot] (6) at (0.75, 1.5) {};
		\node [style=Z phase dot] (7) at (0.75, 2.5) {$\alpha$};
		\node [style=none, rotate=90] (8) at (-0.75, 0) {...};
		\node [style=none, rotate=90] (9) at (1, 0) {...};
	\end{pgfonlayer}
	\begin{pgfonlayer}{edgelayer}
		\draw (0.center) to (1.center);
		\draw (2.center) to (3.center);
		\draw (4) to (6);
		\draw (5) to (6);
		\draw (6) to (7);
	\end{pgfonlayer}
\end{tikzpicture} \ \ ::\ \ |x_1 \ldots x_n \> \mapsto e^{i \alpha \cdot x_1 \oplus \ldots \oplus x_n} |x_1 \ldots x_n\>
\]
Phase gadgets with arbitrary angles applied to arbitrary subsets of qubits form a spanning set for all diagonal unitaries. By restricting phases to integer multiples of $\pi/2^{r-1}$, we obtain a spanning set for all diagonal unitaries on the $r$-th level of the Clifford hierarchy~\cite{cui2017diagonal}. For example, we can represent CCZ, on the 3rd level of the Clifford hierarchy as a collection of $\pi/4$ phase gadgets:
\begin{equation}\label{eq:ccz}
  \text{CCZ} \ = \ \begin{tikzpicture}
	\begin{pgfonlayer}{nodelayer}
		\node [style=none] (0) at (1, 1.5) {};
		\node [style=none] (1) at (7, 1.5) {};
		\node [style=none] (2) at (1, -1.5) {};
		\node [style=none] (3) at (7, -1.5) {};
		\node [style=Z phase dot] (4) at (4, 1.5) {$\frac\pi 4$};
		\node [style=X dot] (6) at (5.25, 0.75) {};
		\node [style=Z phase dot] (7) at (6.25, 0.75) {$\frac\pi 4$};
		\node [style=none] (8) at (1, 0) {};
		\node [style=none] (9) at (7, 0) {};
		\node [style=Z phase dot] (11) at (4, 0) {$\frac\pi 4$};
		\node [style=Z phase dot] (12) at (4, -1.5) {$\frac\pi 4$};
		\node [style=X dot] (13) at (2.75, 0.75) {};
		\node [style=Z phase dot] (14) at (1.75, 0.75) {-$\frac\pi 4$};
		\node [style=X dot] (15) at (2.75, -0.75) {};
		\node [style=Z phase dot] (16) at (1.75, -0.75) {-$\frac\pi 4$};
		\node [style=X dot] (17) at (5.25, -0.75) {};
		\node [style=Z phase dot] (18) at (6.25, -0.75) {-$\frac\pi 4$};
		\node [style=none] (19) at (-11.5, 1.5) {};
		\node [style=none] (20) at (-1, 1.5) {};
		\node [style=none] (21) at (-11.5, 0) {};
		\node [style=none] (22) at (-1, 0) {};
		\node [style=none] (23) at (-11.5, -1.5) {};
		\node [style=none] (24) at (-1, -1.5) {};
		\node [style=none] (25) at (0, 0) {$=$};
		\node [style=X dot] (26) at (-1.5, 2.25) {};
		\node [style=Z phase dot] (27) at (-1.5, 3.25) {$\frac\pi 4$};
		\node [style=Z dot] (28) at (-2.75, 1.5) {};
		\node [style=Z dot] (29) at (-2.75, 0) {};
		\node [style=Z dot] (30) at (-2.75, -1.5) {};
		\node [style=X dot] (31) at (-3.25, 2.25) {};
		\node [style=Z phase dot] (32) at (-3.25, 3.25) {-$\frac\pi 4$};
		\node [style=Z dot] (33) at (-4.5, 1.5) {};
		\node [style=Z dot] (34) at (-4.5, 0) {};
		\node [style=X dot] (35) at (-5, 2.25) {};
		\node [style=Z phase dot] (36) at (-5, 3.25) {-$\frac\pi 4$};
		\node [style=Z dot] (37) at (-6.25, 0) {};
		\node [style=Z dot] (38) at (-6.25, -1.5) {};
		\node [style=X dot] (39) at (-6.5, 2.25) {};
		\node [style=Z phase dot] (40) at (-6.5, 3.25) {-$\frac\pi 4$};
		\node [style=Z dot] (41) at (-7.75, 1.5) {};
		\node [style=Z dot] (42) at (-7.75, -1.5) {};
		\node [style=X dot] (43) at (-8.75, 2.25) {};
		\node [style=Z phase dot] (44) at (-8.75, 3.25) {$\frac\pi 4$};
		\node [style=Z dot] (45) at (-8.75, 1.5) {};
		\node [style=X dot] (46) at (-9.75, 2.25) {};
		\node [style=Z phase dot] (47) at (-9.75, 3.25) {$\frac\pi 4$};
		\node [style=Z dot] (48) at (-9.75, 0) {};
		\node [style=X dot] (49) at (-10.75, 2.25) {};
		\node [style=Z phase dot] (50) at (-10.75, 3.25) {$\frac\pi 4$};
		\node [style=Z dot] (51) at (-10.75, -1.5) {};
	\end{pgfonlayer}
	\begin{pgfonlayer}{edgelayer}
		\draw (0.center) to (1.center);
		\draw (2.center) to (3.center);
		\draw (4) to (6);
		\draw (6) to (7);
		\draw (8.center) to (9.center);
		\draw (11) to (6);
		\draw (12) to (6);
		\draw (13) to (14);
		\draw (13) to (4);
		\draw (13) to (11);
		\draw (15) to (16);
		\draw (17) to (18);
		\draw (4) to (17);
		\draw (12) to (17);
		\draw (15) to (12);
		\draw (15) to (11);
		\draw (19.center) to (20.center);
		\draw (21.center) to (22.center);
		\draw (23.center) to (24.center);
		\draw (26) to (27);
		\draw (26) to (30);
		\draw (29) to (26);
		\draw (28) to (26);
		\draw (31) to (32);
		\draw (34) to (31);
		\draw (33) to (31);
		\draw (35) to (36);
		\draw (38) to (35);
		\draw (37) to (35);
		\draw (39) to (40);
		\draw (42) to (39);
		\draw (41) to (39);
		\draw (43) to (44);
		\draw (45) to (43);
		\draw (46) to (47);
		\draw (48) to (46);
		\draw (49) to (50);
		\draw (51) to (49);
	\end{pgfonlayer}
\end{tikzpicture}
\end{equation}
Note that we have here fused all the Z-spiders on the qubit wires together to write this circuit with multiple phase gadgets more compactly. Each individual gadget corresponds to a term in the associated phase polynomial $\phi : \mathbb B^n \to \mathbb R$, which represents the relative phases of basis elements as $\mathbb R$-linear combinations of parity functions. For example, CCZ can be represented as:
\(
  \text{CCZ} |x_1 x_2 x_3 \> = e^{i \cdot \phi} |x_1 x_2 x_3\>
\)
where $\phi(x_1, x_2, x_3) = \frac\pi4 x_1 + \frac\pi4 x_2 + \frac\pi4 x_3 - \frac\pi4 x_1 \oplus x_2 - \frac\pi4 x_2 \oplus x_3 - \frac\pi4 x_1 \oplus x_3 + \frac\pi4 x_1 \oplus x_2 \oplus x_3$. 

It is also worth noting that phase gadgets applied to the same set of qubits ``fuse'' in the sense that their angles add together, as a consequence of the rules in Figure~\ref{fig:zx-rules}:
\begin{equation}\label{eq:gadget-fusion}
  \begin{tikzpicture}
	\begin{pgfonlayer}{nodelayer}
		\node [style=none] (0) at (-5, 1) {};
		\node [style=none] (1) at (-1, 1) {};
		\node [style=none] (2) at (-5, -1) {};
		\node [style=none] (3) at (-1, -1) {};
		\node [style=Z dot] (4) at (-3.75, 1) {};
		\node [style=Z dot] (5) at (-3.75, -1) {};
		\node [style=X dot] (6) at (-3, 1.5) {};
		\node [style=Z phase dot] (7) at (-3, 2.5) {$\alpha$};
		\node [style=none, rotate=90] (8) at (-4.5, 0) {...};
		\node [style=none, rotate=90] (9) at (-1.5, 0) {...};
		\node [style=none] (10) at (0, 0) {$=$};
		\node [style=Z dot] (11) at (-2.5, 1) {};
		\node [style=Z dot] (12) at (-2.5, -1) {};
		\node [style=X dot] (13) at (-1.75, 1.5) {};
		\node [style=Z phase dot] (14) at (-1.75, 2.5) {$\beta$};
		\node [style=Z dot] (15) at (2, 1) {};
		\node [style=Z dot] (16) at (2, -1) {};
		\node [style=X dot] (17) at (2.75, 1.5) {};
		\node [style=Z phase dot] (18) at (2.75, 2.5) {$\,\alpha\!+\!\beta\,$};
		\node [style=none] (19) at (1, 1) {};
		\node [style=none] (20) at (3.75, 1) {};
		\node [style=none] (21) at (1, -1) {};
		\node [style=none] (22) at (3.75, -1) {};
		\node [style=none, rotate=90] (23) at (1.5, 0) {...};
		\node [style=none, rotate=90] (24) at (3.25, 0) {...};
	\end{pgfonlayer}
	\begin{pgfonlayer}{edgelayer}
		\draw (0.center) to (1.center);
		\draw (2.center) to (3.center);
		\draw (4) to (6);
		\draw (5) to (6);
		\draw (6) to (7);
		\draw (11) to (13);
		\draw (12) to (13);
		\draw (13) to (14);
		\draw (15) to (17);
		\draw (16) to (17);
		\draw (17) to (18);
		\draw (19.center) to (20.center);
		\draw (21.center) to (22.center);
	\end{pgfonlayer}
\end{tikzpicture}
\end{equation}
This is called the \textit{gadget fusion} rule~\cite{kissinger2019tcount}.

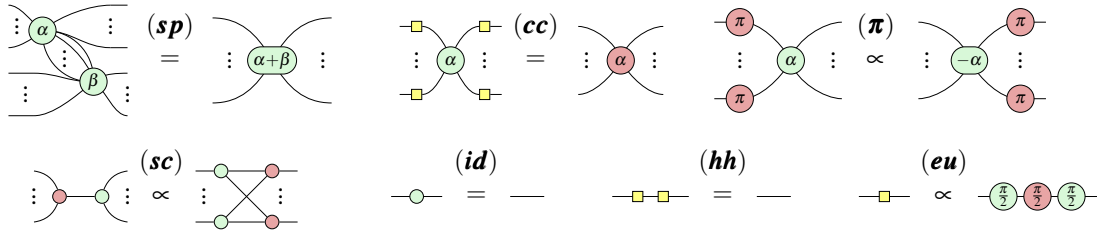
\begin{figure}[!tb]
  \centering
  \begin{tikzpicture}
	\begin{pgfonlayer}{nodelayer}
		\node [style=Z phase dot] (0) at (-17.5, 2.875) {$\beta$};
		\node [style=none] (1) at (-16.5, 3.125) {};
		\node [style=none] (2) at (-19.25, 3.125) {};
		\node [style=none] (3) at (-19.25, 1.875) {};
		\node [style=none] (4) at (-16.5, 1.875) {};
		\node [style=none, rotate=90] (7) at (-19.5, 2.5) {...};
		\node [style=none, rotate=90] (8) at (-16.5, 2.5) {...};
		\node [style=Z phase dot] (10) at (-19, 4.375) {$\alpha$};
		\node [style=none] (11) at (-17, 5.125) {};
		\node [style=none] (12) at (-20, 5.125) {};
		\node [style=none] (13) at (-17, 3.875) {};
		\node [style=none, rotate=90] (14) at (-17, 4.5) {...};
		\node [style=none] (16) at (-20, 3.875) {};
		\node [style=none, rotate=90] (17) at (-19.75, 4.5) {...};
		\node [style=none] (18) at (-16.5, 3.875) {};
		\node [style=none] (19) at (-16.5, 5.125) {};
		\node [style=none] (22) at (-20, 1.875) {};
		\node [style=none] (23) at (-20, 3.125) {};
		\node [style=none] (24) at (-15.25, 3.5) {$=$};
		\node [style=none, rotate=90] (25) at (-18.35, 3.55) {...};
		\node [style=none] (26) at (-14, 4.75) {};
		\node [style=none, rotate=90] (28) at (-13.5, 3.5) {...};
		\node [style=none] (29) at (-10.5, 2.25) {};
		\node [style=none, rotate=90] (30) at (-10.75, 3.5) {...};
		\node [style=Z phase dot] (32) at (-12.25, 3.5) {$\ \alpha\!+\!\beta\ $};
		\node [style=none] (33) at (-10.5, 4.75) {};
		\node [style=none] (34) at (-14, 2.25) {};
		\node [style=none] (35) at (-15.25, 4.5) {\zxspider};
		\node [style=Z phase dot] (36) at (8.25, 3.5) {$\ -\alpha\ $};
		\node [style=none] (37) at (10.5, 4.625) {};
		\node [style=none] (38) at (5.5, 3.5) {\scalareq};
		\node [style=none] (40) at (10.5, 2.375) {};
		\node [style=X phase dot] (42) at (9.75, 2.375) {$\pi$};
		\node [style=none] (44) at (4.5, 4.625) {};
		\node [style=none, rotate=90] (46) at (4.25, 3.5) {...};
		\node [style=none, rotate=90] (48) at (9.75, 3.5) {...};
		\node [style=none] (51) at (4.5, 2.375) {};
		\node [style=X phase dot] (52) at (9.75, 4.625) {$\pi$};
		\node [style=none] (53) at (5.5, 4.5) {\zxpi};
		\node [style=hadamard] (69) at (-8, 4.5) {};
		\node [style=none, rotate=90] (72) at (-8, 3.5) {...};
		\node [style=none] (73) at (-4.5, 3.5) {$=$};
		\node [style=hadamard] (74) at (-8, 2.5) {};
		\node [style=none] (77) at (-8.5, 4.5) {};
		\node [style=none] (81) at (-3.25, 2.5) {};
		\node [style=none] (82) at (-8.5, 2.5) {};
		\node [style=none, rotate=90] (83) at (-3, 3.5) {...};
		\node [style=none] (84) at (-3.25, 4.5) {};
		\node [style=none] (86) at (-4.5, 4.5) {\zxcolor};
		\node [style=Z dot] (87) at (-8, -0.5) {};
		\node [style=none] (88) at (-8.75, -0.5) {};
		\node [style=none] (89) at (-7.25, -0.5) {};
		\node [style=none] (90) at (-5.25, -0.5) {};
		\node [style=none] (91) at (-4.25, -0.5) {};
		\node [style=none] (92) at (-6.25, 0.5) {\zxid};
		\node [style=none] (93) at (-6.25, -0.5) {$=$};
		\node [style=none] (94) at (1, -0.5) {$=$};
		\node [style=none] (95) at (1, 0.5) {\zxhh};
		\node [style=none] (96) at (3, -0.5) {};
		\node [style=none] (97) at (0, -0.5) {};
		\node [style=none] (98) at (-2.25, -0.5) {};
		\node [style=none] (99) at (2, -0.5) {};
		\node [style=hadamard] (100) at (-1.5, -0.5) {};
		\node [style=hadamard] (101) at (-0.75, -0.5) {};
		\node [style=none] (102) at (-15.5, 0.5) {\zxsc};
		\node [style=X dot] (103) at (-12.25, -1.25) {};
		\node [style=none] (104) at (-11.5, 0.25) {};
		\node [style=X dot] (105) at (-18.5, -0.5) {};
		\node [style=Z dot] (106) at (-17.25, -0.5) {};
		\node [style=none] (107) at (-14.5, -1.25) {};
		\node [style=none] (108) at (-19.25, 0.25) {};
		\node [style=none] (109) at (-14.5, 0.25) {};
		\node [style=Z dot] (110) at (-13.75, 0.25) {};
		\node [style=none] (111) at (-16.5, 0.25) {};
		\node [style=none] (112) at (-15.5, -0.5) {\scalareq};
		\node [style=none] (113) at (-19.25, -1.25) {};
		\node [style=X dot] (114) at (-12.25, 0.25) {};
		\node [style=none] (115) at (-11.5, -1.25) {};
		\node [style=Z dot] (116) at (-13.75, -1.25) {};
		\node [style=none] (117) at (-16.5, -1.25) {};
		\node [style=hadamard] (118) at (-6, 4.5) {};
		\node [style=none, rotate=90] (120) at (-6, 3.5) {...};
		\node [style=hadamard] (121) at (-6, 2.5) {};
		\node [style=Z phase dot] (122) at (-7, 3.5) {$\alpha$};
		\node [style=none] (123) at (-5.5, 4.5) {};
		\node [style=none] (126) at (-5.5, 2.5) {};
		\node [style=X phase dot] (127) at (-2, 3.5) {$\alpha$};
		\node [style=none] (128) at (-0.75, 2.5) {};
		\node [style=none, rotate=90] (129) at (-1, 3.5) {...};
		\node [style=none] (130) at (-0.75, 4.5) {};
		\node [style=none] (131) at (7.5, -0.5) {\scalareq};
		\node [style=none] (132) at (7.5, 0.5) {\zxeuler};
		\node [style=none] (133) at (12, -0.5) {};
		\node [style=none] (134) at (6.5, -0.5) {};
		\node [style=none] (135) at (5, -0.5) {};
		\node [style=none] (136) at (8.5, -0.5) {};
		\node [style=hadamard] (137) at (5.75, -0.5) {};
		\node [style=X phase dot] (138) at (10.25, -0.5) {$\frac\pi2$};
		\node [style=Z phase dot] (139) at (9.25, -0.5) {$\frac\pi2$};
		\node [style=Z phase dot] (140) at (11.25, -0.5) {$\frac\pi2$};
		\node [style=Z phase dot] (141) at (3, 3.5) {$\alpha$};
		\node [style=none] (142) at (0.75, 4.625) {};
		\node [style=none] (143) at (0.75, 2.375) {};
		\node [style=X phase dot] (144) at (1.5, 2.375) {$\pi$};
		\node [style=none, rotate=90] (145) at (1.5, 3.5) {...};
		\node [style=X phase dot] (146) at (1.5, 4.625) {$\pi$};
		\node [style=none] (147) at (6.75, 4.625) {};
		\node [style=none, rotate=90] (148) at (7, 3.5) {...};
		\node [style=none] (149) at (6.75, 2.375) {};
		\node [style=none, rotate=90] (150) at (-19.25, -0.5) {...};
		\node [style=none, rotate=90] (151) at (-16.5, -0.5) {...};
		\node [style=none, rotate=90] (152) at (-14.25, -0.5) {...};
		\node [style=none, rotate=90] (153) at (-12, -0.5) {...};
	\end{pgfonlayer}
	\begin{pgfonlayer}{edgelayer}
		\draw [in=-141, out=0, looseness=0.75] (3.center) to (0);
		\draw [in=180, out=-39, looseness=0.75] (0) to (4.center);
		\draw [in=180, out=39, looseness=0.75] (0) to (1.center);
		\draw [in=180, out=0] (2.center) to (0);
		\draw [in=-141, out=0, looseness=0.75] (16.center) to (10);
		\draw [in=180, out=0] (10) to (13.center);
		\draw [in=180, out=39, looseness=0.75] (10) to (11.center);
		\draw [in=141, out=0, looseness=0.75] (12.center) to (10);
		\draw [bend left=45] (10) to (0);
		\draw (11.center) to (19.center);
		\draw (13.center) to (18.center);
		\draw (23.center) to (2.center);
		\draw (22.center) to (3.center);
		\draw [bend left=45] (0) to (10);
		\draw [in=-120, out=0] (34.center) to (32);
		\draw [in=180, out=-60] (32) to (29.center);
		\draw [in=180, out=60] (32) to (33.center);
		\draw [in=120, out=0] (26.center) to (32);
		\draw [in=-60, out=180, looseness=0.75] (42) to (36);
		\draw [in=60, out=180, looseness=0.75] (52) to (36);
		\draw (37.center) to (52);
		\draw (40.center) to (42);
		\draw (77.center) to (69);
		\draw (82.center) to (74);
		\draw (88.center) to (89.center);
		\draw (90.center) to (91.center);
		\draw (98.center) to (97.center);
		\draw (99.center) to (96.center);
		\draw (116) to (114);
		\draw (103) to (110);
		\draw (110) to (114);
		\draw (116) to (103);
		\draw (115.center) to (103);
		\draw (107.center) to (116);
		\draw (110) to (109.center);
		\draw (114) to (104.center);
		\draw [bend right] (113.center) to (105);
		\draw [bend right] (105) to (108.center);
		\draw (105) to (106);
		\draw [bend right] (106) to (117.center);
		\draw [bend left] (106) to (111.center);
		\draw [in=-60, out=180, looseness=0.75] (121) to (122);
		\draw [in=75, out=180, looseness=0.75] (118) to (122);
		\draw (123.center) to (118);
		\draw (126.center) to (121);
		\draw [in=-120, out=0, looseness=0.75] (74) to (122);
		\draw [in=105, out=0, looseness=0.75] (69) to (122);
		\draw [in=180, out=-60, looseness=0.75] (127) to (128.center);
		\draw [in=180, out=75, looseness=0.75] (127) to (130.center);
		\draw [in=0, out=-120, looseness=0.75] (127) to (81.center);
		\draw [in=0, out=105, looseness=0.75] (127) to (84.center);
		\draw [in=105, out=-15, looseness=0.75] (10) to (0);
		\draw (135.center) to (134.center);
		\draw (136.center) to (133.center);
		\draw [in=-120, out=0, looseness=0.75] (144) to (141);
		\draw [in=120, out=0, looseness=0.75] (146) to (141);
		\draw (142.center) to (146);
		\draw (143.center) to (144);
		\draw [in=180, out=-75, looseness=0.75] (141) to (51.center);
		\draw [in=180, out=75, looseness=0.75] (141) to (44.center);
		\draw [in=0, out=-105, looseness=0.75] (36) to (149.center);
		\draw [in=0, out=105, looseness=0.75] (36) to (147.center);
	\end{pgfonlayer}
\end{tikzpicture}
  \caption{The Clifford ZX-calculus: spider fusion \zxspider, colour change \zxcolor, $\pi$-copy \zxpi, strong complementarity \zxsc, identity \zxid, H-cancellation \zxhh, and Euler decomposition of H \zxeuler. Thanks to \zxcolor, all rules hold with their colours reversed.}
  \label{fig:zx-rules}
\end{figure}

\subsection{Spider nests and triorthogonal matrices}\label{sec:triortho-bg}

When non-Clifford angles are assumed to be arbitrary free parameters, we are unlikely to find non-trivial equations between collections of phase gadgets beyond those already provable using the Clifford ZX-calculus (cf.~\cite{vdwetering2024optimal}). However, if we assume that the non-Clifford angles take specific values, especially values of the form $\pi/2^r$ for some integer $r > 1$, more equations hold than just those provable by the Clifford ZX-calculus~\cite{amy2016t}. An important class of such rules are the \textit{spider nest identities}~\cite{debeaudrap2020spidernest2}. These are certain configurations of non-Clifford phase gadgets whose overall action is the identity, up to a global phase.

To understand the collection of all such rules, we should get some structural understanding of what is actually happening. Consider first a single phase gadget and how it acts on the computational basis:
\(
  |x_1 \ldots x_n \> \mapsto e^{i \alpha \cdot x_1 \oplus \ldots \oplus x_n} |x_1 \ldots x_n\>
\).
This action is totally determined by a phase polynomial $\phi : \mathbb F_2^k \to \mathbb R$ here given by $\phi(\vec x)=\alpha \bigoplus_j x_j$. There are two useful bases for expressing phase polynomials: the XOR basis and the monomial basis. The XOR basis represents $\phi$ as a real linear combination of functions of the form $f(\vec x) = \bigoplus_{i \in S} x_i$, where $S \subseteq \{ 1, \ldots, n \}$. The monomial, or ``AND'' basis consists of functions of the form $g(\vec x) = \Pi_{i \in S} x_i$ for some subset $S \subseteq \{ 1, \ldots, n \}$. 

We can transform from the basis of XOR functions into the basis of AND functions by using the fact that $x\oplus y = x+y - 2 x\cdot y$. From this, we can derive the $n$ variable version:
\begin{equation}\label{eq:Boolean-IFourier}
    x_1\oplus\cdots\oplus x_n = \sum_{S \subseteq [n]} (-2)^{|S|-1} \prod_{i \in S} x_i.
\end{equation}
Transforming functions between these two bases is known as the \emph{Boolean Fourier transform}.

Using Eq.~\eqref{eq:Boolean-IFourier} we can rewrite an $\mathbb R$-linear combination of XORs into a linear combination of monomials. Note that the prefactor $(-2)^{|S|-1}$ of each term grows as the degree of the monomial increases. In particular, if the coefficient of $x_1 \oplus \ldots \oplus x_n$ is some integer multiple of $\pi/4$, applying the inverse Fourier transform will result in linear terms whose coefficients are integer multiples of $\pi/4$, quadratic terms with multiples of $\pi/2$, cubic terms with multiples of $\pi$, and all terms of degree 4 or more being multiples of $2\pi$, which will vanish:
\begin{equation}
    e^{i\frac\pi4 x_1\oplus\cdots \oplus x_n} \ =\  \exp \left( i\left(\frac\pi4 \sum_j x_j - \frac\pi2 \sum_{i<j} x_i x_j + \pi \sum_{i<j<k} x_i x_j x_k - 2\pi \cdots \right)\right)
\end{equation}
Hence, each phase gadget corresponds to a collection of $T$ gates (the linear terms), CS gates (quadratic terms) and CCZ terms (cubic terms). A collection of phase gadgets then corresponds to adding together their respective $T$, CS and CCZ gates. We then see that a collection of phase gadgets implements the identity precisely when all these $T$, CS and CCZ gates cancel each other out. This happens when each single variable $x_i$ occurs $0\ \text{mod } 8$ times ($T^8 = I$), each pair of variables occurs $0\ \text{mod\ } 4$ times ($\text{CS}^4 = I$), and each triple occurs $0 \ \text{mod\ } 2$ times ($\text{CCZ}^2 = I$).

We can formalise the cancellation property as follows. Represent a collection of $\pi/4$ phase gadgets as a set of bit strings $\vec y^1, \ldots, \vec y^m$, where $y^l_i = 1$ means the $l$-th phase gadget is connected to the $i$-th wire. Then for the compositions of all $M$ phase gadgets to form an identity they need to satisfy:
\[
    \forall i: \sum_l y_i^l = 0 \md 8 \qquad
    \forall i<j: \sum_l y_i^ly_j^l = 0 \md 4 \qquad
    \forall i<j<k: \sum_l y_i^ly_j^ly_k^l = 0 \md 2
\]
Writing these $\vec y^l$ as the rows of a binary matrix, these conditions specify precisely what it means for the matrix to be \emph{triorthogonal}, namely that each column, product of pairs of columns and product of triples of columns needs to have a Hamming weight that is a multiple of respectively 8, 4 and 2~\cite{nezami2022classification}.

\begin{example}
  The collection of phase gadgets in the LHS of equation \eqref{eq:sn-1-4} corresponds to the following triorthogonal matrix, whose rows are all the non-zero length 4 bit strings:
\begin{equation}\label{eq:sn-1-4-matrix}
    \begin{pmatrix}
        0 & 0 & 0 & 0 & 0 & 0 & 0 & 1 & 1 & 1 & 1 & 1 & 1 & 1 & 1 \\
        0 & 0 & 0 & 1 & 1 & 1 & 1 & 0 & 0 & 0 & 0 & 1 & 1 & 1 & 1 \\
        0 & 1 & 1 & 0 & 0 & 1 & 1 & 0 & 0 & 1 & 1 & 0 & 0 & 1 & 1 \\
        1 & 0 & 1 & 0 & 1 & 0 & 1 & 0 & 1 & 0 & 1 & 0 & 1 & 0 & 1 
      \end{pmatrix}^{T}
\end{equation}
\end{example}

If we instead require the weaker condition that each of these properties holds just modulo $2$, then we get the notion of a \emph{semi-triorthogonal} matrix. In that case, we can abbreviate the 3 conditions into one, where we no longer require $i, j, k$ to be distinct:
\begin{equation}\label{eq:semi-triortho}
    \forall i,j,k: \sum_l y_i^ly_j^ly_k^l = 0 \md 2
\end{equation}

These describe collections of phase gadgets that are equal to a Clifford, instead of exactly equal to the identity. This is because instead of $\sum_l y^l_i = 0 \md 8$, so that each qubit $i$ has a multiple of 8 $T$ gates that cancel out, we just have $\sum_l y^l_i = 0 \md 2$, so that we can pair up all the $T$ gates on the qubit $i$, to combine them as $T^2=S$, which is Clifford. Similarly, the equation $\sum_l y_i^ly_j^l = 0 \md 2$ means that we can pair up all the CS gates on qubits $i$ and $j$ to produce $\text{CS}^2=\text{CZ}$, which is also Clifford. We get the same for the CCZ gates, which pair up into an identity: $\text{CCZ}^2 = I$. 

We summarise the above discussion in the following theorem.
\begin{theorem}\label{thm:gadget-triortho}
  Let $M$ be a boolean matrix with $n$ columns and $k$ rows and define the unitary $U_M$ as the circuit consisting of $k$ $\frac\pi4$ phase gadgets, where the connectivity of the $j$th gadget is described by the $j$th row of $M$. Then $U_M$ is Clifford if and only if $M$ is semi-triorthogonal and $U_M$ is the identity if and only if $M$ is triorthogonal.
\end{theorem}

\begin{remark}
  Several inequivalent definitions of the term ``triorthogonal'' appear in the literature. It is commonly used to describe the weaker condition~\eqref{eq:semi-triortho}, or an even weaker condition that only requires products of pairs and triples of distinct rows to have even Hamming weight. The stronger condition we call \textit{triorthogonal} is also sometimes called 3-even~\cite{vuillot2022quantum}. Our terminology matches e.g.~\cite{Litinski2019gameofsurfacecodes}, with the slight difference that we impose conditions on the columns of a matrix rather than the rows, as it will make some calculations simpler.
\end{remark}

There is clearly a close relationship between the graphical concept of spider nest identities and the (non-graphical) concepts of triorthogonal matrices and low-degree polynomials. To make this formal, it will be useful to have a bridge between the graphical notation and matrices, which thankfully is provided by the scalable ZX-calculus.

\subsection{The scalable ZX-calculus}\label{sec:scalable-zx}

While plain ZX-diagrams are convenient for doing many concrete calculations, it will be convenient when discussing quantum error correcting codes and transversal gates to adopt the \textit{scalable ZX notation}~\cite{SZXCalculus}. This notation enables us to compactly represent operations on registers of many qubits, while still maintaining much of the flavour of calculations with standard ZX-diagrams.

We represent a register of qubits as a single thick wire and the product of $n$ (unconnected) copies of a Z or X spider as a bold spider:
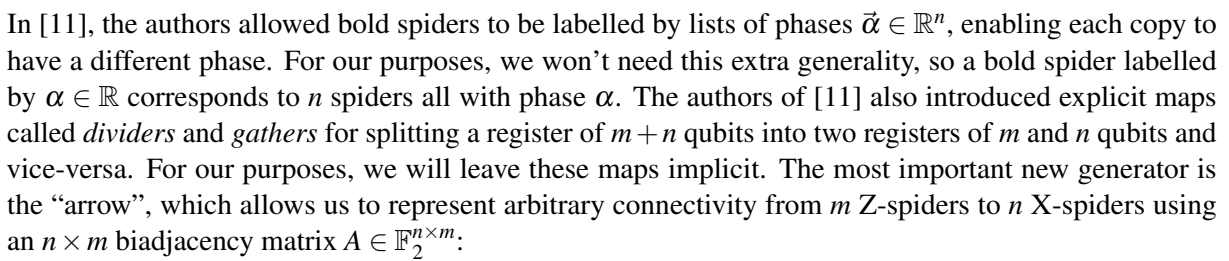
\[
\begin{tikzpicture}
	\begin{pgfonlayer}{nodelayer}
		\node [style=none] (0) at (-1.75, 0) {};
		\node [style=none] (1) at (-0.25, 0) {};
		\node [style=none] (2) at (0.75, 0) {$:=$};
		\node [style=none] (3) at (-1, 0.25) {\footnotesize $n$};
		\node [style=none] (4) at (1.75, 0.5) {};
		\node [style=none] (5) at (1.75, -0.5) {};
		\node [style=none] (6) at (3.25, 0.5) {};
		\node [style=none] (7) at (3.25, -0.5) {};
		\node [style=none, rotate=90] (8) at (2.5, -0.125) {...};
		\node [style=none] (9) at (1.75, 0.25) {};
		\node [style=none] (10) at (3.25, 0.25) {};
		\node [style=none] (11) at (3.75, -0.75) {};
		\node [style=none] (12) at (3.75, 0.75) {};
		\node [style=none] (13) at (4.25, 0) {\footnotesize $n$};
	\end{pgfonlayer}
	\begin{pgfonlayer}{edgelayer}
		\draw [style=boldedge] (0.center) to (1.center);
		\draw (4.center) to (6.center);
		\draw (5.center) to (7.center);
		\draw (9.center) to (10.center);
		\draw [style=brace edge] (12.center) to (11.center);
	\end{pgfonlayer}
\end{tikzpicture}
\qquad\qquad
\begin{tikzpicture}
	\begin{pgfonlayer}{nodelayer}
		\node [style=none] (1) at (-3.25, 0.5) {};
		\node [style=none] (2) at (-3.25, -1) {};
		\node [style=none] (4) at (-0.75, 0.5) {};
		\node [style=none] (5) at (-0.75, -1) {};
		\node [style=none] (17) at (0, 0) {$:=$};
		\node [style=none, rotate=90] (19) at (-1, -0.25) {...};
		\node [style=none, rotate=90] (20) at (-2.75, -0.25) {...};
		\node [style=none] (51) at (6, 1.75) {};
		\node [style=none] (52) at (6, -2.75) {};
		\node [style=none] (58) at (6, 2.75) {};
		\node [style=none] (59) at (6, -1.75) {};
		\node [style=none, rotate=90] (76) at (5.75, -0.75) {...};
		\node [style=none, rotate=90] (79) at (5.75, 2.25) {...};
		\node [style=none, rotate=90] (80) at (5.75, -2.25) {...};
		\node [style=none] (84) at (6, 0.25) {};
		\node [style=none] (85) at (6, 1.25) {};
		\node [style=none, rotate=90] (86) at (5.75, 0.75) {...};
		\node [style=Z phase dot] (87) at (3.5, -0.75) {$\alpha$};
		\node [style=none] (88) at (1, 1.75) {};
		\node [style=none] (89) at (1, -2.75) {};
		\node [style=Z phase dot] (90) at (3.5, 0.75) {$\alpha$};
		\node [style=none] (91) at (1, 2.75) {};
		\node [style=none] (92) at (1, -1.75) {};
		\node [style=none, rotate=90] (93) at (1.25, -0.75) {...};
		\node [style=none, rotate=90] (94) at (1.25, 2.25) {...};
		\node [style=none, rotate=90] (95) at (1.25, -2.25) {...};
		\node [style=none, rotate=90] (96) at (3.5, 0) {...};
		\node [style=none] (97) at (1, 0.25) {};
		\node [style=none] (98) at (1, 1.25) {};
		\node [style=none, rotate=90] (99) at (1.25, 0.75) {...};
		\node [style=none] (100) at (-3.25, 1) {};
		\node [style=Z bold phase dot] (101) at (-2, 0) {$\alpha$};
		\node [style=none] (102) at (-0.75, 1) {};
	\end{pgfonlayer}
	\begin{pgfonlayer}{edgelayer}
		\draw [in=0, out=120] (87) to (88.center);
		\draw [in=0, out=-120] (87) to (89.center);
		\draw [in=0, out=120] (90) to (91.center);
		\draw [in=0, out=-120] (90) to (92.center);
		\draw [in=0, out=150, looseness=0.75] (90) to (98.center);
		\draw [in=0, out=135] (87) to (97.center);
		\draw [in=180, out=60] (87) to (51.center);
		\draw [in=180, out=-60] (87) to (52.center);
		\draw [in=180, out=60] (90) to (58.center);
		\draw [in=180, out=-60] (90) to (59.center);
		\draw [in=-180, out=30, looseness=0.75] (90) to (85.center);
		\draw [in=-180, out=45] (87) to (84.center);
		\draw [style=boldedge, in=180, out=75] (101) to (102.center);
		\draw [style=boldedge, in=0, out=105] (101) to (100.center);
		\draw [style=boldedge, in=180, out=45] (101) to (4.center);
		\draw [style=boldedge, in=180, out=-75] (101) to (5.center);
		\draw [style=boldedge, in=0, out=135] (101) to (1.center);
		\draw [style=boldedge, in=0, out=-105] (101) to (2.center);
	\end{pgfonlayer}
\end{tikzpicture}
\qquad\qquad
\begin{tikzpicture}
	\begin{pgfonlayer}{nodelayer}
		\node [style=none] (1) at (-3.25, 0.5) {};
		\node [style=none] (2) at (-3.25, -1) {};
		\node [style=none] (4) at (-0.75, 0.5) {};
		\node [style=none] (5) at (-0.75, -1) {};
		\node [style=none] (17) at (0, 0) {$:=$};
		\node [style=none, rotate=90] (19) at (-1, -0.25) {...};
		\node [style=none, rotate=90] (20) at (-2.75, -0.25) {...};
		\node [style=none] (51) at (6, 1.75) {};
		\node [style=none] (52) at (6, -2.75) {};
		\node [style=none] (58) at (6, 2.75) {};
		\node [style=none] (59) at (6, -1.75) {};
		\node [style=none, rotate=90] (76) at (5.75, -0.75) {...};
		\node [style=none, rotate=90] (79) at (5.75, 2.25) {...};
		\node [style=none, rotate=90] (80) at (5.75, -2.25) {...};
		\node [style=none] (84) at (6, 0.25) {};
		\node [style=none] (85) at (6, 1.25) {};
		\node [style=none, rotate=90] (86) at (5.75, 0.75) {...};
		\node [style=X phase dot] (87) at (3.5, -0.75) {$\alpha$};
		\node [style=none] (88) at (1, 1.75) {};
		\node [style=none] (89) at (1, -2.75) {};
		\node [style=X phase dot] (90) at (3.5, 0.75) {$\alpha$};
		\node [style=none] (91) at (1, 2.75) {};
		\node [style=none] (92) at (1, -1.75) {};
		\node [style=none, rotate=90] (93) at (1.25, -0.75) {...};
		\node [style=none, rotate=90] (94) at (1.25, 2.25) {...};
		\node [style=none, rotate=90] (95) at (1.25, -2.25) {...};
		\node [style=none, rotate=90] (96) at (3.5, 0) {...};
		\node [style=none] (97) at (1, 0.25) {};
		\node [style=none] (98) at (1, 1.25) {};
		\node [style=none, rotate=90] (99) at (1.25, 0.75) {...};
		\node [style=none] (100) at (-3.25, 1) {};
		\node [style=X bold phase dot] (101) at (-2, 0) {$\alpha$};
		\node [style=none] (102) at (-0.75, 1) {};
	\end{pgfonlayer}
	\begin{pgfonlayer}{edgelayer}
		\draw [in=0, out=120] (87) to (88.center);
		\draw [in=0, out=-120] (87) to (89.center);
		\draw [in=0, out=120] (90) to (91.center);
		\draw [in=0, out=-120] (90) to (92.center);
		\draw [in=0, out=150, looseness=0.75] (90) to (98.center);
		\draw [in=0, out=135] (87) to (97.center);
		\draw [in=180, out=60] (87) to (51.center);
		\draw [in=180, out=-60] (87) to (52.center);
		\draw [in=180, out=60] (90) to (58.center);
		\draw [in=180, out=-60] (90) to (59.center);
		\draw [in=-180, out=30, looseness=0.75] (90) to (85.center);
		\draw [in=-180, out=45] (87) to (84.center);
		\draw [style=boldedge, in=180, out=75] (101) to (102.center);
		\draw [style=boldedge, in=0, out=105] (101) to (100.center);
		\draw [style=boldedge, in=180, out=45] (101) to (4.center);
		\draw [style=boldedge, in=180, out=-75] (101) to (5.center);
		\draw [style=boldedge, in=0, out=135] (101) to (1.center);
		\draw [style=boldedge, in=0, out=-105] (101) to (2.center);
	\end{pgfonlayer}
\end{tikzpicture}
\]
In \cite{SZXCalculus}, the authors allowed bold spiders to be labelled by lists of phases $\vec\alpha \in \mathbb R^n$, enabling each copy to have a different phase. For our purposes, we won't need this extra generality, so a bold spider labelled by $\alpha \in \mathbb R$ corresponds to $n$ spiders all with phase $\alpha$. The authors of \cite{SZXCalculus} also introduced explicit maps called \textit{dividers} and \textit{gathers} for splitting a register of $m+n$ qubits into two registers of $m$ and $n$ qubits and vice-versa. For our purposes, we will leave these maps implicit. The most important new generator is the ``arrow'', which allows us to represent arbitrary connectivity from $m$ Z-spiders to $n$ X-spiders using an $n \times m$ biadjacency matrix $A \in \mathbb F_2^{n\times m}$:
\begin{equation}\label{eq:arrow-def}
  \begin{tikzpicture}
	\begin{pgfonlayer}{nodelayer}
		\node [style=rmat] (0) at (-1.25, 0) {};
		\node [style=none] (1) at (-2.25, 0) {};
		\node [style=none] (2) at (-0.25, 0) {};
		\node [style=none] (3) at (-1.25, 0.625) {\footnotesize $A$};
		\node [style=none] (4) at (0.75, 0) {$:=$};
		\node [style=Z dot] (5) at (2.5, 0.5) {};
		\node [style=Z dot] (6) at (2.5, -0.5) {};
		\node [style=none] (7) at (1.75, 0.5) {};
		\node [style=none] (8) at (1.75, -0.5) {};
		\node [style=none] (9) at (3.25, -0.5) {};
		\node [style=none] (10) at (3.25, -0.125) {};
		\node [style=none] (11) at (3.25, 0.5) {};
		\node [style=none] (12) at (3.25, 0.125) {};
		\node [style=none, rotate=90] (13) at (3.5, -0.5) {...};
		\node [style=none, rotate=90] (14) at (2, 0) {...};
		\node [style=X dot] (15) at (4.5, 0.5) {};
		\node [style=X dot] (16) at (4.5, -0.5) {};
		\node [style=none] (17) at (5.25, 0.5) {};
		\node [style=none] (18) at (5.25, -0.5) {};
		\node [style=none] (19) at (3.75, -0.5) {};
		\node [style=none] (20) at (3.75, -0.125) {};
		\node [style=none] (21) at (3.75, 0.5) {};
		\node [style=none] (22) at (3.75, 0.125) {};
		\node [style=none, rotate=90] (23) at (5, 0) {...};
		\node [style=none] (24) at (3.5, 0) {\footnotesize $A$};
		\node [style=none, rotate=90] (25) at (3.5, 0.5) {...};
	\end{pgfonlayer}
	\begin{pgfonlayer}{edgelayer}
		\draw [style=boldedge] (1.center) to (2.center);
		\draw (7.center) to (5);
		\draw (8.center) to (6);
		\draw (5) to (11.center);
		\draw (5) to (12.center);
		\draw (6) to (10.center);
		\draw (6) to (9.center);
		\draw (17.center) to (15);
		\draw (18.center) to (16);
		\draw (15) to (21.center);
		\draw (15) to (22.center);
		\draw (16) to (20.center);
		\draw (16) to (19.center);
	\end{pgfonlayer}
\end{tikzpicture}
\end{equation}
Taking the convention that $A_i^j$ represents the entry in the $i$-th column and $j$-th row of the matrix $A$, we have in Eq.~\eqref{eq:arrow-def} that $A_i^j = 1$ if and only if the $i$-th Z-spider on the left is connected to the $j$-th X-spider on the right. Concretely, Eq.~\eqref{eq:arrow-def} corresponds, up to scalar factors, to a linear map that acts as $A$ on computational basis vectors:
\[
\begin{tikzpicture}
	\begin{pgfonlayer}{nodelayer}
		\node [style=rmat] (0) at (0, 0) {};
		\node [style=none] (1) at (-1, 0) {};
		\node [style=none] (2) at (1, 0) {};
		\node [style=none] (3) at (0, 0.625) {\footnotesize $A$};
	\end{pgfonlayer}
	\begin{pgfonlayer}{edgelayer}
		\draw [style=boldedge] (1.center) to (2.center);
	\end{pgfonlayer}
\end{tikzpicture}
 \ \ ::\ \ |\vec b\> \mapsto |A \vec b\>
\]
Note that we treat the bitstring $\vec b$ as a column vector for the purposes of matrix multiplication.

Spiders and arrows satisfy several rules that will prove useful. First, we have two ``copy'' laws relating arrows to Z/X spiders:
\begin{equation}\label{eq:arrow-copy}
  \begin{tikzpicture}
	\begin{pgfonlayer}{nodelayer}
		\node [style=rmat] (0) at (-2.75, 0) {};
		\node [style=none] (1) at (-3.75, 0) {};
		\node [style=none] (3) at (-2.75, 0.625) {\footnotesize $A$};
		\node [style=none] (4) at (-0.5, 0.75) {};
		\node [style=Z bold dot] (5) at (-1.75, 0) {};
		\node [style=none] (6) at (-0.5, -0.75) {};
		\node [style=none] (7) at (0.5, 0) {$=$};
		\node [style=rmat] (8) at (4, 0.75) {};
		\node [style=none] (9) at (1.5, 0) {};
		\node [style=none] (10) at (4, 1.375) {\footnotesize $A$};
		\node [style=none] (11) at (3.75, 0.75) {};
		\node [style=Z bold dot] (12) at (2.5, 0) {};
		\node [style=none] (13) at (3.75, -0.75) {};
		\node [style=none] (14) at (4.75, 0.75) {};
		\node [style=none] (15) at (4.75, -0.75) {};
		\node [style=rmat] (16) at (4, -0.75) {};
		\node [style=none] (17) at (4, -0.125) {\footnotesize $A$};
	\end{pgfonlayer}
	\begin{pgfonlayer}{edgelayer}
		\draw [style=boldedge, in=-180, out=-60] (5) to (6.center);
		\draw [style=boldedge] (1.center) to (5);
		\draw [style=boldedge, in=180, out=60] (5) to (4.center);
		\draw [style=boldedge, in=-180, out=-60] (12) to (13.center);
		\draw [style=boldedge] (9.center) to (12);
		\draw [style=boldedge, in=180, out=60] (12) to (11.center);
		\draw [style=boldedge] (11.center) to (14.center);
		\draw [style=boldedge] (13.center) to (15.center);
	\end{pgfonlayer}
\end{tikzpicture}
  \qquad\qquad
  \begin{tikzpicture}
	\begin{pgfonlayer}{nodelayer}
		\node [style=rmat] (0) at (3.75, 0) {};
		\node [style=none] (1) at (4.75, 0) {};
		\node [style=none] (3) at (3.75, 0.625) {\footnotesize $A$};
		\node [style=none] (4) at (1.5, 0.75) {};
		\node [style=X bold dot] (5) at (2.75, 0) {};
		\node [style=none] (6) at (1.5, -0.75) {};
		\node [style=none] (7) at (0.5, 0) {$=$};
		\node [style=rmat] (8) at (-3, 0.75) {};
		\node [style=none] (9) at (-0.5, 0) {};
		\node [style=none] (10) at (-3, 1.375) {\footnotesize $A$};
		\node [style=none] (11) at (-2.75, 0.75) {};
		\node [style=X bold dot] (12) at (-1.5, 0) {};
		\node [style=none] (13) at (-2.75, -0.75) {};
		\node [style=none] (14) at (-3.75, 0.75) {};
		\node [style=none] (15) at (-3.75, -0.75) {};
		\node [style=rmat] (16) at (-3, -0.75) {};
		\node [style=none] (17) at (-3, -0.125) {\footnotesize $A$};
	\end{pgfonlayer}
	\begin{pgfonlayer}{edgelayer}
		\draw [style=boldedge, in=0, out=-120] (5) to (6.center);
		\draw [style=boldedge] (1.center) to (5);
		\draw [style=boldedge, in=0, out=120] (5) to (4.center);
		\draw [style=boldedge, in=0, out=-120] (12) to (13.center);
		\draw [style=boldedge] (9.center) to (12);
		\draw [style=boldedge, in=0, out=120] (12) to (11.center);
		\draw [style=boldedge] (11.center) to (14.center);
		\draw [style=boldedge] (13.center) to (15.center);
	\end{pgfonlayer}
\end{tikzpicture}
\end{equation}

Second, we can express block diagonal matrices in terms of spiders:
\begin{equation}\label{eq:arrow-block-matrix}
\begin{tikzpicture}
	\begin{pgfonlayer}{nodelayer}
		\node [style=rmat] (0) at (-2, 0) {};
		\node [style=none] (1) at (-3, 0) {};
		\node [style=none] (3) at (-2, 1.125) {\footnotesize $\begin{pmatrix}A\\B\end{pmatrix}$};
		\node [style=none] (7) at (0, 0) {$=$};
		\node [style=rmat] (8) at (3.5, 0.75) {};
		\node [style=none] (9) at (1, 0) {};
		\node [style=none] (10) at (3.5, 1.375) {\footnotesize $A$};
		\node [style=none] (11) at (3.25, 0.75) {};
		\node [style=Z bold dot] (12) at (2, 0) {};
		\node [style=none] (13) at (3.25, -0.75) {};
		\node [style=none] (14) at (4.25, 0.75) {};
		\node [style=none] (15) at (4.25, -0.75) {};
		\node [style=rmat] (16) at (3.5, -0.75) {};
		\node [style=none] (17) at (3.5, -0.125) {\footnotesize $B$};
		\node [style=none] (18) at (-1, 0) {};
	\end{pgfonlayer}
	\begin{pgfonlayer}{edgelayer}
		\draw [style=boldedge, in=-180, out=-60] (12) to (13.center);
		\draw [style=boldedge] (9.center) to (12);
		\draw [style=boldedge, in=180, out=60] (12) to (11.center);
		\draw [style=boldedge] (11.center) to (14.center);
		\draw [style=boldedge] (13.center) to (15.center);
		\draw [style=boldedge] (1.center) to (18.center);
	\end{pgfonlayer}
\end{tikzpicture}
\qquad
\begin{tikzpicture}
	\begin{pgfonlayer}{nodelayer}
		\node [style=rmat] (0) at (2.25, 0) {};
		\node [style=none] (1) at (3.5, 0) {};
		\node [style=none] (3) at (2.25, 0.625) {\footnotesize $(A \ \ B)$};
		\node [style=none] (7) at (0, 0) {$=$};
		\node [style=rmat] (8) at (-3.5, 0.75) {};
		\node [style=none] (9) at (-1, 0) {};
		\node [style=none] (10) at (-3.5, 1.375) {\footnotesize $A$};
		\node [style=none] (11) at (-3.25, 0.75) {};
		\node [style=X bold dot] (12) at (-2, 0) {};
		\node [style=none] (13) at (-3.25, -0.75) {};
		\node [style=none] (14) at (-4.25, 0.75) {};
		\node [style=none] (15) at (-4.25, -0.75) {};
		\node [style=rmat] (16) at (-3.5, -0.75) {};
		\node [style=none] (17) at (-3.5, -0.125) {\footnotesize $B$};
		\node [style=none] (18) at (1, 0) {};
	\end{pgfonlayer}
	\begin{pgfonlayer}{edgelayer}
		\draw [style=boldedge, in=0, out=-120] (12) to (13.center);
		\draw [style=boldedge] (9.center) to (12);
		\draw [style=boldedge, in=0, out=120] (12) to (11.center);
		\draw [style=boldedge] (11.center) to (14.center);
		\draw [style=boldedge] (13.center) to (15.center);
		\draw [style=boldedge] (1.center) to (18.center);
	\end{pgfonlayer}
\end{tikzpicture}
\qquad
\begin{tikzpicture}
	\begin{pgfonlayer}{nodelayer}
		\node [style=rmat] (0) at (-2.25, 0) {};
		\node [style=none] (3) at (-2.25, 1.125) {\footnotesize $\begin{pmatrix}A \!\!&\!\! B\\C \!\!&\!\! D\end{pmatrix}$};
		\node [style=none] (7) at (0.5, 0) {$=$};
		\node [style=rmat] (8) at (4.75, 1.9) {};
		\node [style=none] (9) at (2, 0.775) {};
		\node [style=none] (10) at (4.75, 2.525) {\footnotesize $A$};
		\node [style=none] (11) at (4.5, 1.9) {};
		\node [style=Z bold dot] (12) at (3, 0.775) {};
		\node [style=none] (13) at (4.5, -0.6) {};
		\node [style=rmat] (16) at (4.75, -0.6) {};
		\node [style=none] (17) at (4.75, 1.275) {\footnotesize $B$};
		\node [style=none] (18) at (-0.75, 0) {};
		\node [style=none] (25) at (-3.75, 0) {};
		\node [style=rmat] (30) at (4.75, 0.65) {};
		\node [style=none] (31) at (2, -0.725) {};
		\node [style=none] (32) at (4.75, 0.025) {\footnotesize $C$};
		\node [style=none] (33) at (4.5, 0.65) {};
		\node [style=Z bold dot] (34) at (3, -0.725) {};
		\node [style=none] (35) at (4.5, -1.85) {};
		\node [style=none] (36) at (5, 0.65) {};
		\node [style=rmat] (38) at (4.75, -1.85) {};
		\node [style=none] (39) at (4.75, -1.225) {\footnotesize $D$};
		\node [style=none] (40) at (5, -0.6) {};
		\node [style=none] (42) at (7, 1.275) {};
		\node [style=none] (43) at (5, 1.9) {};
		\node [style=X bold dot] (44) at (6, 1.275) {};
		\node [style=none] (45) at (5, 0.65) {};
		\node [style=none] (46) at (7, -1.225) {};
		\node [style=none] (47) at (5, -0.6) {};
		\node [style=X bold dot] (48) at (6, -1.225) {};
		\node [style=none] (49) at (5, -1.85) {};
	\end{pgfonlayer}
	\begin{pgfonlayer}{edgelayer}
		\draw [style=boldedge, in=-180, out=-60] (12) to (13.center);
		\draw [style=boldedge] (9.center) to (12);
		\draw [style=boldedge, in=180, out=60] (12) to (11.center);
		\draw [style=boldedge] (25.center) to (18.center);
		\draw [style=boldedge, in=-180, out=-60] (34) to (35.center);
		\draw [style=boldedge] (31.center) to (34);
		\draw [style=boldedge, in=180, out=60] (34) to (33.center);
		\draw [style=boldedge] (33.center) to (36.center);
		\draw [style=boldedge, in=0, out=-120] (44) to (45.center);
		\draw [style=boldedge] (42.center) to (44);
		\draw [style=boldedge, in=0, out=120] (44) to (43.center);
		\draw [style=boldedge, in=0, out=-120] (48) to (49.center);
		\draw [style=boldedge] (46.center) to (48);
		\draw [style=boldedge, in=0, out=120] (48) to (47.center);
		\draw [style=boldedge] (11.center) to (43.center);
		\draw [style=boldedge] (35.center) to (49.center);
		\draw [style=boldedge] (13.center) to (47.center);
	\end{pgfonlayer}
\end{tikzpicture}
\end{equation}

\section{Inductive construction of spider nest identities}

We can now use the scalable notation to relate spider nest identities to certain families of triorthogonal matrices. The first thing to note is that for any boolean matrix $M$, we can write the associated $n$-qubit diagonal unitary $U_M$ from Theorem~\ref{thm:gadget-triortho}, composed of a $\pi/4$ phase gadget for each of the $k$ rows of $M$, as follows:
\begin{equation}\label{eq:phase-gadget-scalable}
  D_M = \begin{tikzpicture}
	\begin{pgfonlayer}{nodelayer}
		\node [style=none] (0) at (-1, -1) {};
		\node [style=Z bold dot] (1) at (0, -1) {};
		\node [style=none] (2) at (1, -1) {};
		\node [style=Z bold phase dot] (3) at (0, 1.625) {$\frac\pi4$};
		\node [style=umat] (4) at (0, 0.25) {};
		\node [style=none] (5) at (0.75, 0.25) {\footnotesize $M$};
		\node [style=none] (13) at (-0.375, -0.5) {\footnotesize $n$};
		\node [style=none] (14) at (-0.375, 0.875) {\footnotesize $k$};
	\end{pgfonlayer}
	\begin{pgfonlayer}{edgelayer}
		\draw [style=boldedge] (0.center) to (2.center);
		\draw [style=boldedge] (1) to (3);
	\end{pgfonlayer}
\end{tikzpicture}
\end{equation}
$M$ has $n$ columns corresponding to the $n$ qubits of $D_M$ and $k$ rows, corresponding to $k$ phase gadgets. The $i$-th row of $M$ then says which qubits are connected to the $i$-th phase gadget. Hence, a matrix $M$ is triorthogonal if and only if $D_M = I$. We write here $D_M$ instead of $U_M$ to refer to the specific diagram in Eq.~\eqref{eq:phase-gadget-scalable}.

Notably, this gives us an infinite family of graphical equations, of the form $D_M = I$ for all triorthogonal matrices $M$. In fact, this is precisely the set of all spider nest identities, which we justified by the concrete calculations involving the inverse Fourier transform in Section~\ref{sec:triortho-bg}. We know by completeness of the Clifford+T ZX-calculus~\cite{jeandel2019completeness} that all of these equations are provable by an extended version of the ZX-calculus. However, from those rules, it is very difficult to see how one could directly reduce a diagram $D_M$ to $I$ for some fixed $M$, and whether that reduction could be done efficiently (i.e. without expanding to a large normal form before reducing back down). Hence, it is interesting to ask just how much we need to add to the simple Clifford rules in Figure~\ref{fig:zx-rules} in order to prove the entire family of spider nest identities directly. Toward that goal, we will now inductively construct a family of maps that will enable us to generate all the $\pi/4$ spider nest identities.

\begin{definition}\label{def:sn}
  The \emph{spider-nest maps} $s_n : 1 \to n$ are constructed inductively as follows:
  \begin{equation}\label{eq:sn-def}
    \begin{tikzpicture}
	\begin{pgfonlayer}{nodelayer}
		\node [style=growbox] (0) at (-1.25, 0) {$\textrm{s}_n$};
		\node [style=none] (1) at (-2.75, 0) {};
		\node [style=none] (2) at (-0.75, 0) {};
		\node [style=none] (3) at (0.5, 0) {};
		\node [style=none] (7) at (1.5, 0) {$:=$};
		\node [style=growbox] (8) at (6.5, 1.25) {\!$\textrm{s}_{n\textrm{-}1}$\!};
		\node [style=Z dot] (9) at (3.75, 0) {};
		\node [style=none] (11) at (6.75, 1.25) {};
		\node [style=growbox] (15) at (6.5, -1.25) {\!$\textrm{s}_{n\textrm{-}1}$\!};
		\node [style=none] (17) at (6.75, -1.25) {};
		\node [style=X dot] (21) at (5, -1.25) {};
		\node [style=none] (22) at (2.5, 0) {};
		\node [style=none] (23) at (6.5, -2.5) {};
		\node [style=none] (24) at (10.25, -2.5) {};
		\node [style=growbox] (25) at (-9.25, 0) {$\textrm{s}_0$};
		\node [style=none] (26) at (-10.75, 0) {};
		\node [style=none] (27) at (-7.75, 0) {$:=$};
		\node [style=Z phase dot] (28) at (-5.5, 0) {$\frac\pi 4$};
		\node [style=none] (29) at (-6.75, 0) {};
		\node [style=Z bold dot] (30) at (8.75, 0) {};
		\node [style=none] (32) at (10.25, 0) {};
		\node [style=none] (35) at (0, 0.25) {\footnotesize $n$};
		\node [style=none] (36) at (10, 0.375) {\footnotesize $n\textrm{-}1$};
	\end{pgfonlayer}
	\begin{pgfonlayer}{edgelayer}
		\draw (1.center) to (0);
		\draw [style=boldedge] (2.center) to (3.center);
		\draw [in=180, out=60] (9) to (8);
		\draw (21) to (15);
		\draw [in=165, out=-60] (9) to (21);
		\draw (22.center) to (9);
		\draw [in=-180, out=-75] (21) to (23.center);
		\draw (23.center) to (24.center);
		\draw (26.center) to (25);
		\draw (29.center) to (28);
		\draw [style=boldedge, in=120, out=0] (11.center) to (30);
		\draw [style=boldedge, in=-120, out=0] (17.center) to (30);
		\draw [style=boldedge] (30) to (32.center);
	\end{pgfonlayer}
\end{tikzpicture}
  \end{equation}
\end{definition}

Intuitively, this inductive definition results in a phase gadget connecting the single input wire to every subset of the output wires. For example:
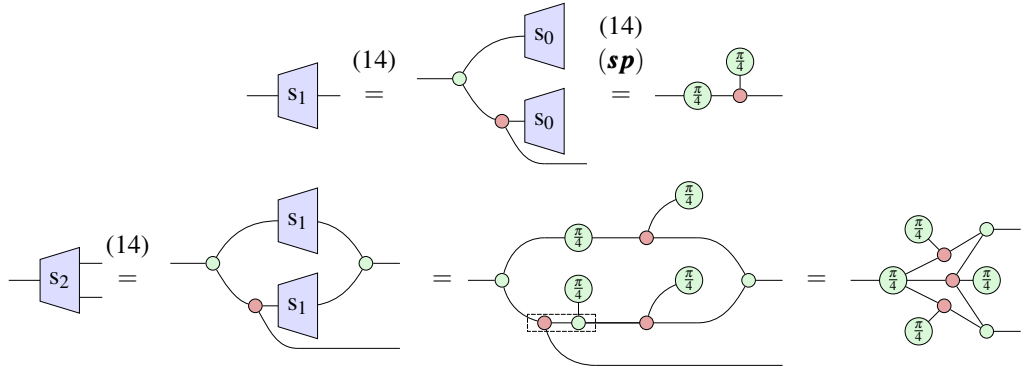
\[\begin{tikzpicture}
	\begin{pgfonlayer}{nodelayer}
		\node [style=growbox] (0) at (-2.25, 0) {$\textrm{s}_1$};
		\node [style=none] (1) at (-3.75, 0) {};
		\node [style=none] (2) at (-1, 0) {};
		\node [style=none] (3) at (-2, 0) {};
		\node [style=none] (6) at (0, 0) {$=$};
		\node [style=growbox] (7) at (5, 1.75) {$\textrm{s}_{0}$};
		\node [style=Z dot] (8) at (2.5, 0.5) {};
		\node [style=growbox] (10) at (5, -0.75) {$\textrm{s}_{0}$};
		\node [style=X dot] (12) at (3.75, -0.75) {};
		\node [style=none] (13) at (1.25, 0.5) {};
		\node [style=none] (14) at (5, -2) {};
		\node [style=none] (15) at (6.25, -2) {};
		\node [style=none] (16) at (7.25, 0) {$=$};
		\node [style=Z phase dot] (18) at (9.5, 0) {$\frac\pi4$};
		\node [style=Z phase dot] (19) at (10.75, 1) {$\frac\pi4$};
		\node [style=X dot] (20) at (10.75, 0) {};
		\node [style=none] (21) at (8.25, 0) {};
		\node [style=none] (23) at (12, 0) {};
		\node [style=none] (24) at (0, 1) {\eqref{eq:sn-def}};
		\node [style=none] (25) at (7.25, 2) {\eqref{eq:sn-def}};
		\node [style=none] (26) at (7.25, 1) {\zxspider};
	\end{pgfonlayer}
	\begin{pgfonlayer}{edgelayer}
		\draw (1.center) to (0);
		\draw (3.center) to (2.center);
		\draw [in=180, out=60] (8) to (7);
		\draw (12) to (10);
		\draw [in=165, out=-60, looseness=0.75] (8) to (12);
		\draw (13.center) to (8);
		\draw [in=-180, out=-60] (12) to (14.center);
		\draw (14.center) to (15.center);
		\draw (20) to (19);
		\draw (18) to (20);
		\draw (21.center) to (18);
		\draw (20) to (23.center);
	\end{pgfonlayer}
\end{tikzpicture}\]
\[\begin{tikzpicture}
	\begin{pgfonlayer}{nodelayer}
		\node [style=growbox] (0) at (-2.25, 0) {$\textrm{s}_2$};
		\node [style=none] (1) at (-3.75, 0) {};
		\node [style=none] (2) at (-1, 0.5) {};
		\node [style=none] (3) at (-2, 0.5) {};
		\node [style=none] (4) at (-2, -0.5) {};
		\node [style=none] (5) at (-1, -0.5) {};
		\node [style=none] (6) at (-0.25, 0) {$=$};
		\node [style=growbox] (7) at (4.75, 1.75) {$\textrm{s}_{1}$};
		\node [style=Z dot] (8) at (2.25, 0.5) {};
		\node [style=none] (9) at (5, 1.75) {};
		\node [style=growbox] (10) at (4.75, -0.75) {$\textrm{s}_{1}$};
		\node [style=none] (11) at (5, -0.75) {};
		\node [style=X dot] (12) at (3.5, -0.75) {};
		\node [style=none] (13) at (1, 0.5) {};
		\node [style=none] (14) at (4.75, -2) {};
		\node [style=none] (15) at (7.75, -2) {};
		\node [style=Z dot] (16) at (6.75, 0.5) {};
		\node [style=none] (17) at (7.75, 0.5) {};
		\node [style=none] (18) at (9, 0) {$=$};
		\node [style=Z phase dot] (19) at (13, 1.25) {$\frac\pi4$};
		\node [style=Z dot] (20) at (10.75, 0) {};
		\node [style=none] (21) at (16.5, 1.25) {};
		\node [style=none] (22) at (16.5, -1.25) {};
		\node [style=X dot] (24) at (12, -1.25) {};
		\node [style=none] (25) at (9.75, 0) {};
		\node [style=none] (26) at (13.5, -2.5) {};
		\node [style=none] (27) at (19, -2.5) {};
		\node [style=Z dot] (28) at (18, 0) {};
		\node [style=none] (29) at (19, 0) {};
		\node [style=Z phase dot] (31) at (16.25, 2.5) {$\frac\pi4$};
		\node [style=X dot] (32) at (15, 1.25) {};
		\node [style=Z dot] (33) at (13, -1.25) {};
		\node [style=Z phase dot] (35) at (16.25, 0) {$\frac\pi4$};
		\node [style=X dot] (36) at (15, -1.25) {};
		\node [style=none] (37) at (12, 1.25) {};
		\node [style=none] (57) at (20, 0) {$=$};
		\node [style=none] (63) at (21, 0) {};
		\node [style=Z dot] (66) at (25, 1.5) {};
		\node [style=none] (67) at (26, 1.5) {};
		\node [style=Z phase dot] (76) at (22.25, 0) {$\frac\pi4$};
		\node [style=Z dot] (77) at (25, -1.5) {};
		\node [style=none] (78) at (26, -1.5) {};
		\node [style=X dot] (79) at (23.75, 0.75) {};
		\node [style=X dot] (80) at (24, 0) {};
		\node [style=X dot] (81) at (23.75, -0.75) {};
		\node [style=Z phase dot] (82) at (25, 0) {$\frac\pi4$};
		\node [style=Z phase dot] (83) at (23, 1.5) {$\frac\pi4$};
		\node [style=Z phase dot] (84) at (23, -1.5) {$\frac\pi4$};
		\node [style=none] (85) at (-0.25, 1) {\eqref{eq:sn-def}};
		\node [style=Z phase dot] (86) at (13, -0.25) {$\frac\pi4$};
		\node [style=none] (87) at (11.5, -1.5) {};
		\node [style=none] (88) at (13.5, -1.5) {};
		\node [style=none] (89) at (13.5, -1) {};
		\node [style=none] (90) at (11.5, -1) {};
	\end{pgfonlayer}
	\begin{pgfonlayer}{edgelayer}
		\draw (1.center) to (0);
		\draw (3.center) to (2.center);
		\draw (4.center) to (5.center);
		\draw [in=180, out=60] (8) to (7);
		\draw (12) to (10);
		\draw [in=165, out=-60, looseness=0.75] (8) to (12);
		\draw (13.center) to (8);
		\draw [in=-180, out=-60] (12) to (14.center);
		\draw (14.center) to (15.center);
		\draw [in=120, out=0] (9.center) to (16);
		\draw [in=-120, out=0] (11.center) to (16);
		\draw (16) to (17.center);
		\draw (24) to (22.center);
		\draw [in=165, out=-75] (20) to (24);
		\draw (25.center) to (20);
		\draw [in=-180, out=-75] (24) to (26.center);
		\draw (26.center) to (27.center);
		\draw [in=120, out=0] (21.center) to (28);
		\draw (28) to (29.center);
		\draw [in=-165, out=75] (32) to (31);
		\draw (33) to (36);
		\draw [in=-165, out=75] (36) to (35);
		\draw [in=-120, out=0] (22.center) to (28);
		\draw [in=-180, out=75] (20) to (37.center);
		\draw (37.center) to (21.center);
		\draw (66) to (67.center);
		\draw (63.center) to (76);
		\draw (77) to (78.center);
		\draw (76) to (80);
		\draw (80) to (66);
		\draw (80) to (82);
		\draw (80) to (77);
		\draw (76) to (79);
		\draw (79) to (66);
		\draw (79) to (83);
		\draw (76) to (81);
		\draw (81) to (77);
		\draw (81) to (84);
		\draw (33) to (86);
		\draw [style=dashed edge] (89.center)
			 to (90.center)
			 to (87.center)
			 to (88.center)
			 to cycle;
	\end{pgfonlayer}
\end{tikzpicture}\]
where the last step follows from applying the strong complementarity rule \zxsc to the marked spider pair, and then applying spider fusion \zxspider as much as possible.

We now formalise this intuitive explanation of $s_n$ using scalable notation. Let $B_n$ be the $n\times 2^n$ matrix whose $2^n$ rows consist of all $n$-bitstrings. That is, the matrix defined inductively as follows:
\[
  B_0 = ()
  \qquad\qquad
  B_n = \begin{pmatrix}
    B_{n-1} & \vec 0 \\
    B_{n-1} & \vec 1
  \end{pmatrix}
\]
where $\vec 0$ and $\vec 1$ are respectively the column vectors of all $0$'s and all $1$'s. For example, we have:
\begin{equation}
  B_1 \ = \ \begin{pmatrix}
    0 \\ 1
  \end{pmatrix} 
  \qquad \quad 
  B_2 \ = \ \begin{pmatrix}
    0 & 0 \\ 1 & 0 \\
    0 & 1 \\ 1 & 1
  \end{pmatrix}
\end{equation}

\begin{theorem}\label{thm:sn-char}
  For all $n$, we have:
  \begin{equation}\label{eq:sn-char}
    \begin{tikzpicture}
	\begin{pgfonlayer}{nodelayer}
		\node [style=growbox] (0) at (-2.25, 0) {$s_n$};
		\node [style=none] (1) at (-3.5, 0) {};
		\node [style=none] (2) at (-1, 0) {};
		\node [style=none] (3) at (-0.25, 0) {$=$};
		\node [style=none] (4) at (0.5, 0) {};
		\node [style=none] (5) at (1.5, 0) {};
		\node [style=none] (6) at (1.5, 0.75) {\footnotesize $\vec 1$};
		\node [style=X bold dot] (7) at (2.5, 0) {};
		\node [style=lmat] (8) at (3.5, 0) {};
		\node [style=none] (9) at (3.5, 0.75) {\footnotesize $B_n$};
		\node [style=none] (10) at (4.5, 0) {};
		\node [style=Z bold phase dot] (11) at (2.5, 1) {$\frac\pi4$};
		\node [style=rmat] (12) at (1.5, 0) {};
	\end{pgfonlayer}
	\begin{pgfonlayer}{edgelayer}
		\draw (1.center) to (0);
		\draw [style=boldedge] (0) to (2.center);
		\draw [style=boldedge] (5.center) to (7);
		\draw (4.center) to (5.center);
		\draw [style=boldedge] (7) to (10.center);
		\draw [style=boldedge] (7) to (11);
	\end{pgfonlayer}
\end{tikzpicture}
  \end{equation}
\end{theorem}
\begin{proof}
  First, note that:
  \begin{equation}\label{eq:bitstring-rep-pf}
    \begin{tikzpicture}
	\begin{pgfonlayer}{nodelayer}
		\node [style=rmat] (0) at (-2.5, 0) {};
		\node [style=none] (1) at (-2.5, 0.625) {\footnotesize $B_n$};
		\node [style=none] (2) at (0, 0) {$=$};
		\node [style=rmat] (3) at (3.5, 1.9) {};
		\node [style=none] (4) at (0.75, 0.775) {};
		\node [style=none] (5) at (3.5, 2.525) {\footnotesize $B_{n-1}$};
		\node [style=none] (6) at (3.25, 1.9) {};
		\node [style=Z bold dot] (7) at (1.75, 0.775) {};
		\node [style=none] (8) at (3.25, -0.6) {};
		\node [style=rmat] (9) at (3.5, -0.6) {};
		\node [style=none] (10) at (3.5, 1.275) {\footnotesize $\vec 0$};
		\node [style=none] (11) at (-1, 0) {};
		\node [style=none] (12) at (-4, 0) {};
		\node [style=rmat] (13) at (3.5, 0.65) {};
		\node [style=none] (14) at (0.75, -0.725) {};
		\node [style=none] (15) at (3.5, 0.025) {\footnotesize $B_{n-1}$};
		\node [style=Z dot] (17) at (1.75, -0.725) {};
		\node [style=none] (18) at (3.5, -1.85) {};
		\node [style=rmat] (20) at (3.5, -1.85) {};
		\node [style=none] (21) at (3.5, -1.225) {\footnotesize $\vec 1$};
		\node [style=none] (22) at (3.75, -0.6) {};
		\node [style=none] (23) at (5.75, 1.275) {};
		\node [style=none] (24) at (3.75, 1.9) {};
		\node [style=X bold dot] (25) at (4.75, 1.275) {};
		\node [style=none] (26) at (3.5, 0.65) {};
		\node [style=none] (27) at (5.75, -1.225) {};
		\node [style=none] (28) at (3.75, -0.6) {};
		\node [style=X bold dot] (29) at (4.75, -1.225) {};
		\node [style=none] (31) at (7, 0) {$=$};
		\node [style=none] (58) at (13.75, 0) {$=$};
		\node [style=none] (60) at (21.75, 1.025) {};
		\node [style=Z bold dot] (63) at (24.25, 1.025) {};
		\node [style=rmat] (65) at (23.25, 1.025) {};
		\node [style=none] (67) at (23.25, 1.65) {\footnotesize $B_{n-1}$};
		\node [style=rmat] (70) at (23.25, -0.975) {};
		\node [style=none] (71) at (23.25, -0.35) {\footnotesize $\vec 1$};
		\node [style=none] (76) at (21.75, -0.975) {};
		\node [style=none] (80) at (26.5, -1) {};
		\node [style=X bold dot] (82) at (25.75, -1) {};
		\node [style=none] (83) at (26.5, 1) {};
		\node [style=rmat] (84) at (10.75, 1.9) {};
		\node [style=none] (85) at (8, 0.775) {};
		\node [style=none] (86) at (10.75, 2.525) {\footnotesize $B_{n-1}$};
		\node [style=none] (87) at (10.5, 1.9) {};
		\node [style=Z bold dot] (88) at (9, 0.775) {};
		\node [style=none] (89) at (10.5, -0.6) {};
		\node [style=rmat] (90) at (10.75, -0.6) {};
		\node [style=none] (93) at (8, -0.725) {};
		\node [style=none] (94) at (10.75, 0.025) {\footnotesize $B_{n-1}$};
		\node [style=Z dot] (95) at (10.25, 0.65) {};
		\node [style=Z dot] (96) at (9, -0.725) {};
		\node [style=none] (97) at (10.75, -1.85) {};
		\node [style=rmat] (99) at (10.75, -1.85) {};
		\node [style=none] (100) at (10.75, -1.225) {\footnotesize $\vec 1$};
		\node [style=none] (101) at (11, -0.6) {};
		\node [style=none] (102) at (13, 1.275) {};
		\node [style=none] (103) at (11, 1.9) {};
		\node [style=X bold dot] (104) at (12, 1.275) {};
		\node [style=X bold dot] (105) at (11.25, 0.65) {};
		\node [style=none] (106) at (13, -1.225) {};
		\node [style=none] (107) at (11, -0.6) {};
		\node [style=X bold dot] (108) at (12, -1.225) {};
		\node [style=none] (109) at (20.75, 0) {$=$};
		\node [style=rmat] (110) at (17.5, 1) {};
		\node [style=none] (111) at (14.75, 1.025) {};
		\node [style=none] (112) at (17.5, 1.625) {\footnotesize $B_{n-1}$};
		\node [style=none] (113) at (17.25, 1) {};
		\node [style=Z bold dot] (114) at (15.75, 1.025) {};
		\node [style=none] (115) at (17.25, -0.1) {};
		\node [style=rmat] (116) at (17.5, -0.1) {};
		\node [style=none] (117) at (14.75, -1.225) {};
		\node [style=none] (118) at (17.5, 0.525) {\footnotesize $B_{n-1}$};
		\node [style=none] (121) at (17.5, -1.575) {};
		\node [style=rmat] (122) at (17.5, -1.575) {};
		\node [style=none] (123) at (17.5, -0.95) {\footnotesize $\vec 1$};
		\node [style=none] (124) at (17.75, -0.1) {};
		\node [style=none] (125) at (19.75, 1) {};
		\node [style=none] (126) at (17.75, 1) {};
		\node [style=none] (129) at (19.75, -0.975) {};
		\node [style=none] (130) at (17.75, -0.1) {};
		\node [style=X bold dot] (131) at (18.75, -0.975) {};
		\node [style=none] (132) at (7, 1) {\eqref{eq:arrow-def}};
		\node [style=none] (133) at (20.75, 1) {\eqref{eq:arrow-copy}};
		\node [style=none] (134) at (13.75, 2) {\zxspider};
		\node [style=none] (135) at (13.75, 1) {\idrule};
	\end{pgfonlayer}
	\begin{pgfonlayer}{edgelayer}
		\draw [style=boldedge, in=-180, out=-60] (7) to (8.center);
		\draw [style=boldedge] (4.center) to (7);
		\draw [style=boldedge, in=180, out=60] (7) to (6.center);
		\draw [style=boldedge] (12.center) to (11.center);
		\draw [in=-180, out=-60] (17) to (18.center);
		\draw (14.center) to (17);
		\draw [style=boldedge, in=0, out=-120] (25) to (26.center);
		\draw [style=boldedge] (23.center) to (25);
		\draw [style=boldedge, in=0, out=120] (25) to (24.center);
		\draw [style=boldedge] (27.center) to (29);
		\draw [style=boldedge, in=0, out=120] (29) to (28.center);
		\draw [style=boldedge] (6.center) to (24.center);
		\draw [style=boldedge] (8.center) to (28.center);
		\draw [style=boldedge] (60.center) to (63);
		\draw [style=boldedge] (82) to (63);
		\draw [style=boldedge] (63) to (83.center);
		\draw (76.center) to (70);
		\draw [style=boldedge] (70) to (82);
		\draw [style=boldedge] (82) to (80.center);
		\draw [style=boldedge, in=0, out=-120] (29) to (18.center);
		\draw [style=boldedge, in=-180, out=-60] (88) to (89.center);
		\draw [style=boldedge] (85.center) to (88);
		\draw [style=boldedge, in=180, out=60] (88) to (87.center);
		\draw [in=-180, out=-60] (96) to (97.center);
		\draw (93.center) to (96);
		\draw [in=180, out=60] (96) to (95);
		\draw [style=boldedge, in=0, out=-120] (104) to (105);
		\draw [style=boldedge] (102.center) to (104);
		\draw [style=boldedge, in=0, out=120] (104) to (103.center);
		\draw [style=boldedge] (106.center) to (108);
		\draw [style=boldedge, in=0, out=120] (108) to (107.center);
		\draw [style=boldedge] (87.center) to (103.center);
		\draw [style=boldedge] (89.center) to (107.center);
		\draw [style=boldedge, in=0, out=-120] (108) to (97.center);
		\draw [in=180, out=60] (17) to (26.center);
		\draw [style=boldedge, in=-180, out=-60] (114) to (115.center);
		\draw [style=boldedge] (111.center) to (114);
		\draw [style=boldedge] (114) to (113.center);
		\draw [style=boldedge] (129.center) to (131);
		\draw [style=boldedge, in=0, out=120] (131) to (130.center);
		\draw [style=boldedge] (113.center) to (126.center);
		\draw [style=boldedge] (115.center) to (130.center);
		\draw [style=boldedge, in=0, out=-135] (131) to (121.center);
		\draw (126.center) to (125.center);
		\draw [in=180, out=0] (117.center) to (122);
	\end{pgfonlayer}
\end{tikzpicture}
  \end{equation}
  
  Using this equation and the scalable rules, we can prove \eqref{eq:sn-char} from \eqref{eq:sn-def} by induction on $n$:
  \ctikzfig{sn-char-induc-pf}
  Here in the last step we gathered together the wires connecting to the X-spiders. This works because the matrix arrows with $\vec 1$ are just a single Z-spider fully-connected to identity X-spiders on the right. These fuse with the surrounding Z- and X-spiders to produce the right result.
\end{proof}

If we connect this $s_n$ generator to an X-spider on the left, we obtain the following:
\[
  \begin{tikzpicture}
	\begin{pgfonlayer}{nodelayer}
		\node [style=growbox] (0) at (-2.5, 0) {$s_n$};
		\node [style=none] (1) at (-4.75, 0) {};
		\node [style=none] (2) at (-1.25, 0) {};
		\node [style=none] (3) at (-0.5, 0) {$=$};
		\node [style=none] (5) at (2.25, 0) {};
		\node [style=none] (6) at (2.25, 0.75) {\footnotesize $\vec 1$};
		\node [style=X bold dot] (7) at (3.5, 0) {};
		\node [style=lmat] (8) at (4.5, 0) {};
		\node [style=none] (9) at (4.5, 0.75) {\footnotesize $B_n$};
		\node [style=none] (10) at (5.5, 0) {};
		\node [style=Z bold phase dot] (11) at (3.5, 1) {$\frac\pi4$};
		\node [style=rmat] (12) at (2.25, 0) {};
		\node [style=rmat] (14) at (-3.75, 0) {};
		\node [style=none] (15) at (-3.7, 0.75) {\footnotesize $\vec 1^T$};
		\node [style=none] (17) at (0.25, 0) {};
		\node [style=rmat] (18) at (1.25, 0) {};
		\node [style=none] (19) at (1.3, 0.75) {\footnotesize $\vec 1^T$};
		\node [style=none] (20) at (-3.75, 0) {};
		\node [style=none] (21) at (1.25, 0) {};
		\node [style=growbox] (22) at (-8, 0) {$s_n$};
		\node [style=none] (23) at (-7.75, 0.75) {};
		\node [style=none] (24) at (-6.25, 0.75) {};
		\node [style=none] (25) at (-7.75, -0.75) {};
		\node [style=none] (26) at (-6.25, -0.75) {};
		\node [style=none, rotate=90] (27) at (-6.75, 0) {...};
		\node [style=none] (28) at (-9.75, 0.5) {};
		\node [style=none] (29) at (-10.5, 0.5) {};
		\node [style=none] (30) at (-9.75, -0.5) {};
		\node [style=none] (31) at (-10.5, -0.5) {};
		\node [style=none, rotate=90] (32) at (-10.25, 0) {...};
		\node [style=X dot] (33) at (-9.25, 0) {};
		\node [style=none] (34) at (-5.5, 0) {$=$};
		\node [style=none] (35) at (7, 0) {};
		\node [style=none] (36) at (8, 0.75) {\footnotesize $\bm 1$};
		\node [style=X bold dot] (37) at (9, 0) {};
		\node [style=lmat] (38) at (10, 0) {};
		\node [style=none] (39) at (10, 0.75) {\footnotesize $B_n$};
		\node [style=none] (40) at (11, 0) {};
		\node [style=Z bold phase dot] (41) at (9, 1) {$\frac\pi4$};
		\node [style=rmat] (42) at (8, 0) {};
		\node [style=none] (47) at (6.25, 0) {$=$};
	\end{pgfonlayer}
	\begin{pgfonlayer}{edgelayer}
		\draw [style=boldedge] (0) to (2.center);
		\draw [style=boldedge] (5.center) to (7);
		\draw [style=boldedge] (7) to (10.center);
		\draw [style=boldedge] (7) to (11);
		\draw [style=boldedge] (1.center) to (20.center);
		\draw (20.center) to (0);
		\draw [style=boldedge] (17.center) to (21.center);
		\draw (21.center) to (5.center);
		\draw (23.center) to (24.center);
		\draw (25.center) to (26.center);
		\draw (28.center) to (29.center);
		\draw (30.center) to (31.center);
		\draw (33) to (22);
		\draw [in=0, out=105] (33) to (28.center);
		\draw [in=0, out=-105] (33) to (30.center);
		\draw [style=boldedge] (35.center) to (37);
		\draw [style=boldedge] (37) to (40.center);
		\draw [style=boldedge] (37) to (41);
	\end{pgfonlayer}
\end{tikzpicture}
\]
where $\bm 1$ is the matrix where every entry is 1. We will see in the next section how this map lets us generate all of the spider nest identities.

\section{Proving all spider nest identities}

In this section, we will show that, by adding just one rule to the Clifford ZX-calculus, we can prove all spider nest equations. We call this one extra equation the \textit{S4-rule}:
\begin{equation}\label{eq:s4-rule}
    \begin{tikzpicture}
	\begin{pgfonlayer}{nodelayer}
		\node [style=growbox] (0) at (-2.75, 0) {$\textrm{s}_4$};
		\node [style=none] (1) at (-4.25, 0) {};
		\node [style=none] (2) at (-2.25, 0.75) {};
		\node [style=none] (3) at (-1, 0.75) {};
		\node [style=none] (4) at (-2.25, -0.25) {};
		\node [style=none] (5) at (-1, -0.25) {};
		\node [style=none] (7) at (0, 0) {$=$};
		\node [style=none] (8) at (-2.25, 0.25) {};
		\node [style=none] (9) at (-1, 0.25) {};
		\node [style=none] (10) at (-2.25, -0.75) {};
		\node [style=none] (11) at (-1, -0.75) {};
		\node [style=Z dot] (12) at (1.75, 0) {};
		\node [style=Z dot] (13) at (2.75, 0.75) {};
		\node [style=none] (14) at (3.5, 0.75) {};
		\node [style=Z dot] (15) at (2.75, -0.25) {};
		\node [style=none] (16) at (3.5, -0.25) {};
		\node [style=Z dot] (17) at (2.75, 0.25) {};
		\node [style=none] (18) at (3.5, 0.25) {};
		\node [style=Z dot] (19) at (2.75, -0.75) {};
		\node [style=none] (20) at (3.5, -0.75) {};
		\node [style=none] (21) at (1, 0) {};
	\end{pgfonlayer}
	\begin{pgfonlayer}{edgelayer}
		\draw (1.center) to (0);
		\draw (2.center) to (3.center);
		\draw (4.center) to (5.center);
		\draw (8.center) to (9.center);
		\draw (10.center) to (11.center);
		\draw (13) to (14.center);
		\draw (15) to (16.center);
		\draw (17) to (18.center);
		\draw (19) to (20.center);
		\draw (21.center) to (12);
	\end{pgfonlayer}
\end{tikzpicture}
\end{equation}
We will show this is sound, but we will first need some more definitions.
In particular, we will start by passing to an alternative characterisation for (semi-)triorthogonal matrices, in terms of polynomials of low degree. This result seems to be well-known, and appears in various guises in the literature, e.g. when applying Reed-Muller codes to T-count minimisation or magic state distillation. The version we give here is a variation on one given by Nezami and Haah~\cite{nezami2022classification}.

\begin{definition}
  For a boolean matrix $M$ with $n$ columns, its \textit{indicator polynomial} $P_M \in \mathbb F_2[x_1, \ldots, x_n]$ sends a bitstring $(b_1, \ldots, b_n)$ to $1$ if and only if that bitstring appears as a row in $M$ an odd number of times.
\end{definition}

\begin{theorem}\label{thm:indicator}
  A matrix $M$ with $n$ columns is semi-triorthogonal if and only if its indicator polynomial $P_M$ is of degree at most $n - 4$.
\end{theorem}
The proof is straightforward, but relies on some basic facts about Reed-Muller codes. We give these and the proof of Theorem~\ref{thm:indicator} in Appendix~\ref{app:indicator}.

We can now show that S4 is indeed sound. First, note that it is equivalent up to the standard Clifford rules to the following rule:
\begin{equation}\label{eq:s4-as-gadgets}
  \begin{tikzpicture}
	\begin{pgfonlayer}{nodelayer}
		\node [style=growbox] (0) at (-4.25, 1) {$\textrm{s}_4$};
		\node [style=none] (2) at (-3.75, 1.75) {};
		\node [style=Z dot] (3) at (-2.5, -0.25) {};
		\node [style=none] (4) at (-3.75, 0.75) {};
		\node [style=Z dot] (5) at (-2.5, -2.25) {};
		\node [style=none] (6) at (0, 0) {$=$};
		\node [style=none] (7) at (-3.75, 1.25) {};
		\node [style=Z dot] (8) at (-2.5, -1.25) {};
		\node [style=none] (9) at (-3.75, 0.25) {};
		\node [style=Z dot] (10) at (-2.5, -3.25) {};
		\node [style=none] (21) at (-6, 2.25) {};
		\node [style=none] (22) at (-1.5, 2.25) {};
		\node [style=none] (23) at (-6, -0.25) {};
		\node [style=none] (24) at (-1.5, -0.25) {};
		\node [style=none] (25) at (-6, -1.25) {};
		\node [style=none] (26) at (-1.5, -1.25) {};
		\node [style=none] (27) at (-6, -2.25) {};
		\node [style=none] (28) at (-1.5, -2.25) {};
		\node [style=none] (29) at (-6, -3.25) {};
		\node [style=none] (30) at (-1.5, -3.25) {};
		\node [style=Z dot] (31) at (-5.5, 2.25) {};
		\node [style=none] (41) at (1.75, 2.25) {};
		\node [style=none] (42) at (3.75, 2.25) {};
		\node [style=none] (43) at (1.75, -0.25) {};
		\node [style=none] (44) at (3.75, -0.25) {};
		\node [style=none] (45) at (1.75, -1.25) {};
		\node [style=none] (46) at (3.75, -1.25) {};
		\node [style=none] (47) at (1.75, -2.25) {};
		\node [style=none] (48) at (3.75, -2.25) {};
		\node [style=none] (49) at (1.75, -3.25) {};
		\node [style=none] (50) at (3.75, -3.25) {};
	\end{pgfonlayer}
	\begin{pgfonlayer}{edgelayer}
		\draw [in=120, out=0] (2.center) to (3);
		\draw [in=135, out=0, looseness=0.50] (4.center) to (5);
		\draw [in=135, out=0, looseness=0.75] (7.center) to (8);
		\draw [in=135, out=0, looseness=0.50] (9.center) to (10);
		\draw (21.center) to (22.center);
		\draw (23.center) to (24.center);
		\draw (25.center) to (26.center);
		\draw (27.center) to (28.center);
		\draw (29.center) to (30.center);
		\draw [in=-180, out=-30] (31) to (0);
		\draw (41.center) to (42.center);
		\draw (43.center) to (44.center);
		\draw (45.center) to (46.center);
		\draw (47.center) to (48.center);
		\draw (49.center) to (50.center);
	\end{pgfonlayer}
\end{tikzpicture}
\end{equation}
Here the left-hand side consists of $2^4=16$ phase gadgets, which are all connected to the first qubit, and have all possible connections to the bottom four qubits. Its connectivity matrix is hence $M=(\vec 1\ \ B_4)$. A bitstring $x_1x_2x_3x_4x_5$ hence appears as a row in $M$ iff $x_1=1$. Its indicator polynomial is then $P(x_1,x_2,x_3,x_4,x_5) = x_1$. This is a degree 1 polynomial on 5 variables, and hence satisfies the condition of Theorem~\ref{thm:indicator}, so that $M$ is semi-triorthogonal. We can also manually verify that $M$ is in fact triorthogonal, so that Eq.~\eqref{eq:s4-as-gadgets} is indeed correct.

Returning to the S4 rule, we see that, thanks to the inductive definition of $s_n$, not only does $s_4$ separate, but so do all $s_n$ for $n \geq 4$.

\begin{lemma}\label{lem:sn-sep}
  For $n \geq 4$, the Clifford ZX-calculus augmented with the S4 rule implies:
  \begin{equation}\label{eq:sn-sep}
    \begin{tikzpicture}
	\begin{pgfonlayer}{nodelayer}
		\node [style=growbox] (0) at (-2.75, 0) {$s_n$};
		\node [style=none] (1) at (-4.25, 0) {};
		\node [style=none] (2) at (-1.25, 0) {};
		\node [style=none] (3) at (-1.5, 0.5) {\footnotesize $n$};
		\node [style=none] (4) at (0, 0) {$=$};
		\node [style=none] (6) at (1, 0) {};
		\node [style=none] (7) at (4, 0) {};
		\node [style=none] (8) at (3.75, 0.5) {\footnotesize $n$};
		\node [style=Z dot] (9) at (2, 0) {};
		\node [style=Z bold dot] (10) at (3, 0) {};
	\end{pgfonlayer}
	\begin{pgfonlayer}{edgelayer}
		\draw (1.center) to (0);
		\draw [style=boldedge] (0) to (2.center);
		\draw [style=boldedge] (10) to (7.center);
		\draw (6.center) to (9);
	\end{pgfonlayer}
\end{tikzpicture}
  \end{equation}
\end{lemma}
\begin{proof}
  By induction on $n$. The base case $n=4$ is \eqref{eq:s4-rule}. For $n > 4$, unfold $s_n$ using \eqref{eq:sn-def} and apply the induction hypothesis.
\end{proof}

We see then that if we include $s_n$ for $n\geq 4$ in a circuit, that we obtain:
\begin{equation}\label{eq:sn-x-spiders}
    \begin{tikzpicture}
	\begin{pgfonlayer}{nodelayer}
		\node [style=none] (0) at (1, 1) {};
		\node [style=none] (1) at (7.25, 1) {};
		\node [style=none] (2) at (1, -1.25) {};
		\node [style=none] (3) at (7.25, -1.25) {};
		\node [style=none] (4) at (1, -3) {};
		\node [style=none] (5) at (7.25, -3) {};
		\node [style=Z dot] (6) at (6.5, -1.25) {};
		\node [style=Z dot] (7) at (6.5, -3) {};
		\node [style=growbox] (8) at (3.75, -0.125) {$\textrm{s}_n$};
		\node [style=Z dot] (9) at (1.75, 1) {};
		\node [style=none] (10) at (4.25, 0.375) {};
		\node [style=none] (11) at (4.25, -0.625) {};
		\node [style=none] (80) at (1, 2.75) {};
		\node [style=none] (81) at (7.25, 2.75) {};
		\node [style=none] (82) at (1, 2.75) {};
		\node [style=Z dot] (83) at (1.75, 2.75) {};
		\node [style=X dot] (84) at (2.625, -0.125) {};
		\node [style=none, rotate=90] (85) at (1.75, 1.875) {...};
		\node [style=none] (104) at (0, 0) {$=$};
		\node [style=none] (106) at (-5.75, -0.5) {};
		\node [style=none] (107) at (-5.5, 0.25) {\footnotesize $\bm 1$};
		\node [style=X bold dot] (108) at (-4.5, -0.5) {};
		\node [style=lmat] (109) at (-3.5, -0.5) {};
		\node [style=none] (110) at (-3.5, 0.25) {\footnotesize $B_n$};
		\node [style=none] (111) at (-3.25, -0.5) {};
		\node [style=Z bold phase dot] (112) at (-4.5, 0.5) {$\frac\pi4$};
		\node [style=rmat] (113) at (-5.5, -0.5) {};
		\node [style=Z bold dot] (114) at (-7.25, 1.25) {};
		\node [style=none] (115) at (-8, 1.25) {};
		\node [style=none] (116) at (-1.25, 1.25) {};
		\node [style=Z bold dot] (117) at (-2.5, -1.25) {};
		\node [style=none] (118) at (-8, -1.25) {};
		\node [style=none] (119) at (-1.25, -1.25) {};
		\node [style=none] (120) at (-1.5, 1.625) {\footnotesize $k$};
		\node [style=none] (121) at (-1.575, -0.875) {\footnotesize $n$};
		\node [style=none] (122) at (9.25, 1) {};
		\node [style=none] (123) at (15.5, 1) {};
		\node [style=none] (124) at (9.25, -1.25) {};
		\node [style=none] (125) at (15.5, -1.25) {};
		\node [style=none] (126) at (9.25, -3) {};
		\node [style=none] (127) at (15.5, -3) {};
		\node [style=Z dot] (128) at (14.75, -1.25) {};
		\node [style=Z dot] (129) at (14.75, -3) {};
		\node [style=Z dot] (130) at (11.5, -0.125) {};
		\node [style=Z dot] (131) at (10, 1) {};
		\node [style=Z dot] (132) at (12.5, 0.375) {};
		\node [style=Z dot] (133) at (12.5, -0.625) {};
		\node [style=none] (135) at (9.25, 2.75) {};
		\node [style=none] (136) at (15.5, 2.75) {};
		\node [style=none] (137) at (9.25, 2.75) {};
		\node [style=Z dot] (138) at (10, 2.75) {};
		\node [style=X dot] (139) at (10.875, -0.125) {};
		\node [style=none] (142) at (8.25, 0) {$=$};
		\node [style=none, rotate=90] (143) at (6.5, -2.075) {...};
		\node [style=none, rotate=90] (144) at (10, 1.875) {...};
		\node [style=none, rotate=90] (145) at (14.75, -2.125) {...};
		\node [style=none, rotate=90] (146) at (12.5, -0.125) {...};
		\node [style=none, rotate=90] (147) at (4.5, -0.125) {...};
		\node [style=none] (148) at (17.5, 1) {};
		\node [style=none] (149) at (20, 1) {};
		\node [style=none] (150) at (17.5, -1.25) {};
		\node [style=none] (151) at (20, -1.25) {};
		\node [style=none] (152) at (17.5, -3) {};
		\node [style=none] (153) at (20, -3) {};
		\node [style=none] (160) at (17.5, 2.75) {};
		\node [style=none] (161) at (20, 2.75) {};
		\node [style=none] (162) at (17.5, 2.75) {};
		\node [style=none] (165) at (16.5, 0) {$=$};
		\node [style=none, rotate=90] (166) at (18.75, 1.875) {...};
		\node [style=none, rotate=90] (167) at (18.75, -2.125) {...};
	\end{pgfonlayer}
	\begin{pgfonlayer}{edgelayer}
		\draw (0.center) to (1.center);
		\draw (2.center) to (3.center);
		\draw (4.center) to (5.center);
		\draw [in=150, out=0] (11.center) to (7);
		\draw [in=150, out=0, looseness=0.75] (10.center) to (6);
		\draw (82.center) to (81.center);
		\draw [in=105, out=-60] (83) to (84);
		\draw [in=-75, out=150, looseness=0.75] (84) to (9);
		\draw (84) to (8);
		\draw [style=boldedge] (106.center) to (108);
		\draw [style=boldedge] (108) to (111.center);
		\draw [style=boldedge] (108) to (112);
		\draw [style=boldedge] (115.center) to (116.center);
		\draw [style=boldedge, in=180, out=-90] (114) to (106.center);
		\draw [style=boldedge] (118.center) to (119.center);
		\draw [style=boldedge, in=90, out=0] (111.center) to (117);
		\draw (122.center) to (123.center);
		\draw (124.center) to (125.center);
		\draw (126.center) to (127.center);
		\draw [in=150, out=0] (133) to (129);
		\draw [in=150, out=0, looseness=0.75] (132) to (128);
		\draw (137.center) to (136.center);
		\draw [in=105, out=-60] (138) to (139);
		\draw [in=-75, out=150, looseness=0.75] (139) to (131);
		\draw (139) to (130);
		\draw (148.center) to (149.center);
		\draw (150.center) to (151.center);
		\draw (152.center) to (153.center);
		\draw (162.center) to (161.center);
	\end{pgfonlayer}
\end{tikzpicture}
\end{equation}

The lefthand-side above consists of the set of all phase gadgets that are connected to the first $k$ wires and all combinations of the last $n$ wires. Writing these phase gadgets as the rows of a matrix, this is $M = (\bm 1\ \ B_n)$. The indicator polynomial of this matrix $P$ satisfies the condition that $P(x_1, \ldots, x_{n+k}) = 1$ if and only if $x_j = 1$ for $j \in \{1, \ldots, k\}$. Hence, it is the monomial $P = x_1\ldots x_k$. By permuting wires, we can obtain any monomial on $m = k+n$ variables, and as long as the degree $k \leq m - 4$, then equation \eqref{eq:sn-x-spiders} is satisfied. Hence, in particular, if $M$ is a boolean matrix with $m$ columns and no duplicate rows whose indicator polynomial is of degree at most $m-4$, then the Clifford rules together with the S4 rule proves:
\begin{equation}\label{eq:triortho-map2}
  \begin{tikzpicture}
	\begin{pgfonlayer}{nodelayer}
		\node [style=none] (0) at (-3, 0) {};
		\node [style=Z bold dot] (1) at (-2, 0) {};
		\node [style=none] (2) at (-1, 0) {};
		\node [style=Z bold phase dot] (3) at (-2, 1.875) {$\frac\pi4$};
		\node [style=umat] (4) at (-2, 0.75) {};
		\node [style=none] (5) at (-1.25, 0.75) {\footnotesize $M$};
		\node [style=none] (10) at (0, 0) {$=$};
		\node [style=none] (11) at (1, 0) {};
		\node [style=none] (12) at (2.5, 0) {};
	\end{pgfonlayer}
	\begin{pgfonlayer}{edgelayer}
		\draw [style=boldedge] (0.center) to (2.center);
		\draw [style=boldedge] (1) to (3);
		\draw [style=boldedge] (11.center) to (12.center);
	\end{pgfonlayer}
\end{tikzpicture}
\end{equation}

\begin{theorem}\label{thm:diag-complete}
  The Clifford ZX-calculus, plus the S4 rule \eqref{eq:s4-rule} are \textit{diagonally complete} for Clifford+T ZX-diagrams. That is, for any boolean matrices $M, N$, if $\llbracket D_M \rrbracket = \llbracket D_N \rrbracket$, the $D_M$ can be transformed into $D_N$ using only the rules in Figure~\ref{fig:zx-rules} and S4.
\end{theorem}
\begin{proof}
  Proving $D_M = D_N$ is equivalent to proving $D_MD_N^\dagger = I$. First note that $D_N^\dagger$ is equal to $D_N$ up to a Clifford diagram, as $\frac\pi4$ phase gadgets and $-\frac\pi4$ phase gadgets differ by a Clifford. The circuit for $D_MD_N^\dagger$ consists of a combination of phase gadgets, which can be combined into a single matrix describing the connectivity of the gadgets. Duplicate rows can always be cancelled by applying the gadget fusion rule \eqref{eq:gadget-fusion}. Hence, to prove completeness it suffices to prove that when $\llbracket D_M \rrbracket = \llbracket I \rrbracket$ we can rewrite $D_M$ into $I$.
  Concretely, $D_M = I$ is true if and only if $M$ is triorthogonal. Let $P_{M}$ be its indicator polynomial, which will have degree at most $n-4$ for $n$ the number of qubits (Theorem~\ref{thm:indicator}). Let $N_1, \ldots, N_k$ be matrices whose indicator monomials are $P_1, \ldots, P_k$ for $P_{M} = \sum_j P_j$. Since these all have degree $\leq n-4$, we can show using Eq.~\eqref{eq:triortho-map2} that:
  \[
 \begin{tikzpicture}
	\begin{pgfonlayer}{nodelayer}
		\node [style=none] (0) at (-8, 0) {};
		\node [style=Z bold dot] (1) at (-7, 0) {};
		\node [style=none] (2) at (-6, 0) {};
		\node [style=Z bold phase dot] (3) at (-7, 1.875) {$\frac\pi4$};
		\node [style=umat] (4) at (-7, 0.75) {};
		\node [style=none] (5) at (-6.25, 0.75) {\footnotesize $M$};
		\node [style=none] (10) at (-5, 0) {$=$};
		\node [style=none] (11) at (-4, 0) {};
		\node [style=Z bold dot] (12) at (-3, 0) {};
		\node [style=none] (13) at (3.25, 0) {};
		\node [style=Z bold phase dot] (14) at (-3, 1.875) {$\frac\pi4$};
		\node [style=umat] (15) at (-3, 0.75) {};
		\node [style=none] (16) at (-2.25, 0.75) {\footnotesize $M$};
		\node [style=Z bold dot] (17) at (-0.75, 0) {};
		\node [style=Z bold phase dot] (18) at (-0.75, 1.875) {$\frac\pi4$};
		\node [style=umat] (19) at (-0.75, 0.75) {};
		\node [style=none] (20) at (0, 0.75) {\footnotesize $N_1$};
		\node [style=none] (21) at (0.5, 1.5) {...};
		\node [style=Z bold dot] (22) at (1.75, 0) {};
		\node [style=Z bold phase dot] (23) at (1.75, 1.875) {$\frac\pi4$};
		\node [style=umat] (24) at (1.75, 0.75) {};
		\node [style=none] (25) at (2.5, 0.75) {\footnotesize $N_k$};
		\node [style=none] (26) at (4.25, 0) {$=$};
		\node [style=none] (27) at (5.25, 0) {};
		\node [style=Z bold dot] (28) at (6.25, 0) {};
		\node [style=none] (29) at (7.25, 0) {};
		\node [style=Z bold phase dot] (30) at (6.25, 2.375) {$\frac\pi4$};
		\node [style=umat] (31) at (6.25, 1.25) {};
		\node [style=none] (32) at (8.375, 1.25) {\footnotesize $L := \begin{pmatrix} M \\ N_1 \\ \vdots \\ N_k \end{pmatrix}$};
		\node [style=none] (33) at (-5, 1) {\eqref{eq:triortho-map2}};
		\node [style=none] (34) at (4.25, 1) {\eqref{eq:arrow-block-matrix}};
		\node [style=none] (35) at (4.25, 2) {\zxspider};
	\end{pgfonlayer}
	\begin{pgfonlayer}{edgelayer}
		\draw [style=boldedge] (0.center) to (2.center);
		\draw [style=boldedge] (1) to (3);
		\draw [style=boldedge] (11.center) to (13.center);
		\draw [style=boldedge] (12) to (14);
		\draw [style=boldedge] (17) to (18);
		\draw [style=boldedge] (22) to (23);
		\draw [style=boldedge] (27.center) to (29.center);
		\draw [style=boldedge] (28) to (30);
	\end{pgfonlayer}
\end{tikzpicture}   
  \]
  Then, the indicator polynomial of $L$ is $P_M + \sum_j P_j = 0$. Hence, every row in $L$ appears an even number of times. Using gadget-fusion, we can therefore reduce all angles to integer multiples of $\pi/2$. We can then apply Clifford-completeness to reduce to the identity.
\end{proof}

\begin{remark}
Note that while S4 makes the ZX-calculus diagonally complete for Clifford+T ZX-diagrams, this rule set is not complete for all Clifford+T diagrams. To see this, one can check that S4 is sound for the $\llbracket-\rrbracket^\sharp$ interpretation given in~\cite{supplementarity} and hence cannot derive e.g.~the supplementarity law for non-Clifford angles.\footnote{The authors wish to thank Richie Yeung for pointing this out.}
\end{remark}

One way to think of the Clifford+S4 ZX-calculus is as the ZX analogue of the equational presentation for CNOT-Dihedral circuits of Amy \textit{et al}~\cite{amy2017dihedral}. Indeed we can get the following as a corollary of Theorem~\ref{thm:diag-complete}.

\begin{corollary}\label{cor:CNOT+T}
  The Clifford rules plus S4 are complete for CNOT+T circuits.
\end{corollary}
\begin{proof}
  For CNOT+T circuits $U, V$ where $\llbracket U \rrbracket = \llbracket V \rrbracket$, it suffices to show that $U^\dagger V$ can be rewritten to the identity. Using just Clifford rules, it is possible to rewrite any CNOT+T circuit into a layer of phase gadgets $D_M$ followed by a CNOT circuit $C$. Since $\llbracket U^\dagger V \rrbracket = I$, it must be the case the $\llbracket D_M \rrbracket = \llbracket C \rrbracket = I$. Hence, we can use the completeness of phase-free ZX to rewrite $C$ into $I$ and Theorem~\ref{thm:diag-complete} to rewrite $D_M$ into $I$.
\end{proof}

A Clifford+T circuit can be written as a composition of CNOT, $T$ and Hadamard gates. Hence, the results above apply to those Clifford+T circuits without Hadamards. For any circuit with Hadamards, we can ``split it up'' at those points and reason about the CNOT+$T$ circuits separately. This technique is used for instance when doing $T$-count optimisation~\cite{debeaudrap2020spidernest2,heyfron2018efficient,ruiz2024quantum}, where the circuit is either split up on the Hadamards, or the Hadamards are made into CZ gates by the use of an ancilla and measurement~\cite{heyfron2018efficient}.
Note that there do seem to be Clifford+T circuit rewrite rules that specifically involve Hadamard gates~\cite{coecke2018zx} that seem likely to not be provable just using Clifford rules and spider-nest identities.

\begin{remark}
  Theorem~\ref{thm:diag-complete} and Corollary~\ref{cor:CNOT+T} should be compared to the results in~\cite{deBeaudrap2020Techniques,Munson2019ANDgates}. There they show that assuming the 15 $T$ gate spider-nest Eq.~\eqref{eq:sn-1-4} they can prove the family of spider nests consisting of a single high-degree phase gadget combined with many phase gadgets of degree 3 and lower. 
  However, they only use these rules for circuit optimisation and do not discuss completeness.
\end{remark}

\section{Characterisation of transversal $\mathcal D_3$ gates for CSS codes}

We can use our representation of spider-nest identities in the scalable ZX-calculus to prove a characterisation of CSS codes with transversal diagonal logical operations in the third level of the Clifford hierarchy.

A \emph{CSS code} is a stabiliser code that has a generating set of stabilisers consisting of pure $X$ stabilisers and pure $Z$ stabilisers (i.e.~that are respectively tensor products of $I$ and $X$, or $I$ and $Z$). Its logical $X$ (resp. $Z$) operators also purely consist of $X$ (resp. $Z$) operators. The $Z$ stabilisers and logical operators are in fact completely determined by the $X$ stabilisers and logical operators (up to choice of generator), and hence we only have to specify the ``$X$ part'' of the CSS code.

We can hence fully specify a $\llbracket n, k, d \rrbracket$ CSS code---i.e.~one that has $n$ physical qubits, $k$ logical qubits, and is distance $d$--- by fixing $k+r \leq n$ linearly independent vectors in $\mathbb F_2^n$ corresponding to $k$ logical X operators and $r$ X stabilisers. Here a $1$ at position $i$ in the vector denotes an $X$ on the $i$th qubit~\cite{kissinger2022grok}. Letting $L$ and $S$ be the matrices that have these vectors as their columns ($L$ standing for ``logical'' and $S$ for ``stabiliser''), the vectors corresponding to the remaining $n-k-r$ $Z$ checks can be recovered by fixing a basis for $\textrm{span}(\textrm{cols}(L) \cup \textrm{cols}(S))^\perp$.
Following~\cite{kissinger2022grok}, we can write the encoder for a CSS code as a phase-free ZX-diagram. If we choose to write it in Z-X normal form, the encoder consists of a row of Z-spiders at the input and a row of X-spiders at the output. For each logical operator, an input wire connects to a Z-spider, which then connects to X-spiders on the output corresponding to the support of the operator. For each $X$ stabiliser, an additional Z-spider with no inputs connects to output X-spiders, according to the support of the operator. While this is a bit unwieldy to say in words, the encoder can be written in terms of the matrices $L$ and $S$ straightforwardly as follows:
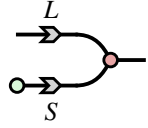
\ctikzfig{encoder-scalable}
Indeed this corresponds to the linear map that sends basis vectors $|\vec b\> \in (\mathbb C^2)^{\otimes k}$ of the logical space to their associated codewords $\sum_{\vec c}|L\vec b + S\vec c\> \in (\mathbb C^2)^{\otimes n}$.

Note that diagonal unitaries of the form $D_P$
correspond to transversal applications of powers of $T$ gates if and only if the rows of $P$ each have Hamming weight $1$ (as $T$ gates are phase gadgets connected to just a single qubit). With this, we are now ready to state our characterisation result.

\begin{theorem}\label{thm:char-transversal}
  A CSS code with X-logical operators and X-stabilisers $L$ and $S$ admits a transversal implementation of a gate $D_H^\dagger \in \mathcal D_3$ if and only if there exists a matrix $P$ whose rows have Hamming weight $1$ such that the matrix
  \[
    M = \left(\begin{array}{c|c}
      H & 0 \\ \hline
      PL & PS
    \end{array}\right)
    \]
    is triorthogonal.
\end{theorem}

\begin{proof}
  First, we will apply the scalable ZX rules to graphically decompose the block form of $M$:
  \ctikzfig{triortho-pf-sm}
  For $E$ the encoder associated with the CSS code $(L, S)$, the code implements $D_H^\dagger$ if and only if for some $P$ we have $D_P E D_H = E$. We begin by ``pushing'' $D_P$ through the encoder:
  \ctikzfig{transversal-push}
  If $M$ is triorthogonal, the part marked $M$ in the diagram above will vanish and we are left with the encoder $E$. If it is not triorthogonal, then $D_M \neq I$. Since $D_M$ is diagonal, it follows that $D_M |+\ldots+\> \neq |+\ldots+\>$, so $D_P E D_H |+\ldots+\> \neq E |+\ldots+\>$. Hence, $D_P E D_H \neq E$.
\end{proof}

\begin{remark}
  For simplicity, in Theorem~\ref{thm:char-transversal} we took the meaning of ``transversal'' to be ``consisting of single-qubit operations''. However, there is a more general notion of transversality that classifies operations e.g. between two or more copies of a code. For instance, if a CSS code is the tensor product of two identical CSS codes, then implementing pairwise CNOT gates between the corresponding qubits of the two codes will implement a logical transversal CNOT. This kind of operation is still called transversal because each CNOT gate only involves a single qubit in each copy of the code, i.e. in each \textit{code block}. If our code actually consists of $b$ distinct code blocks, we can accomodate this more general notion of transversality by replacing the condition that the rows of $P$ have Hamming weight $1$ by a condition that $P$ can be factored across code blocks as $P = (\,P_1 \ \ P_2 \ \ \cdots \ \ P_b\,)$ where each of the $P_i$ has rows of hamming weight at most $1$.
\end{remark}

Note if we take $P = I$ (corresponding to a single $T$ gate on each physical qubit), then from any triorthogonal matrix of the form:
\[
    M = \left(\begin{array}{c|c}
      H & 0 \\ \hline
      L & S
    \end{array}\right)
\]
we can read off a CSS code $(L,S)$ with a transversal implementation of $D_H^\dagger$. Many notable CSS codes with transversal gates can be seen as instances of this construction. For the examples below we write $RM(k,n)$ for the \emph{Reed-Muller} code space of polynomials on $n$ variables of degree at most $k$.

\begin{example}
The degree-1 monomial $x_1 \in \RM(1, 5)$ gives the following triorthogonal matrix:
\[
\left(
\begin{array}{c|ccccccccccccccc}
  1 & 1 & 1 & 1 & 1 & 1 & 1 & 1 & 1 & 1 & 1 & 1 & 1 & 1 & 1 & 1 \\ \hline
  0 & 1 & 0 & 1 & 0 & 1 & 0 & 1 & 0 & 1 & 0 & 1 & 0 & 1 & 0 & 1 \\
  0 & 0 & 1 & 1 & 0 & 0 & 1 & 1 & 0 & 0 & 1 & 1 & 0 & 0 & 1 & 1 \\
  0 & 0 & 0 & 0 & 1 & 1 & 1 & 1 & 0 & 0 & 0 & 0 & 1 & 1 & 1 & 1 \\
  0 & 0 & 0 & 0 & 0 & 0 & 0 & 0 & 1 & 1 & 1 & 1 & 1 & 1 & 1 & 1
\end{array}    
\right)^T
\]
The matrices $(L, S)$ on the right define a $\llbracket 15, 1, 3 \rrbracket$ quantum Reed-Muller code. Reading $H$ from the top-left, we see that it has a transversal implementation of $T^\dagger$. This property is used as the basis of the original 15-to-1 magic state distillation protocol given by Bravyi and Kitaev~\cite{bravyi2005universal}.
\end{example}

\begin{example}
The constant polynomial $1 \in \RM(0,4)$ gives us a triorthogonal matrix whose columns are all the 4-bitstrings. If we partition the matrix as follows:
\[
\left(
\begin{array}{cccccccc|cccccccc}
  0 & 1 & 0 & 1 & 0 & 1 & 0 & 1 & 0 & 1 & 0 & 1 & 0 & 1 & 0 & 1 \\
  0 & 0 & 1 & 1 & 0 & 0 & 1 & 1 & 0 & 0 & 1 & 1 & 0 & 0 & 1 & 1 \\
  0 & 0 & 0 & 0 & 1 & 1 & 1 & 1 & 0 & 0 & 0 & 0 & 1 & 1 & 1 & 1 \\ \hline
  0 & 0 & 0 & 0 & 0 & 0 & 0 & 0 & 1 & 1 & 1 & 1 & 1 & 1 & 1 & 1
\end{array}    
\right)^T
\]
then $(L, S)$ defines the X-logical operators and X-stabiliers of an $\llbracket 8, 3, 2 \rrbracket$ colour code, dubbed the ``smallest interesting colour code''~\cite{campbell2016smallest}. In the upper left, we see the phase gadgets of a CCZ gate as in equation~\eqref{eq:ccz}, hence this code admits a transversal $\textrm{CCZ}^\dagger = \textrm{CCZ}$.
\end{example}

While Theorem~\ref{thm:char-transversal} gives a complete characterisation for the transversal gates in a CSS code, it is not obvious from the statement whether it can be used to efficiently find such gates. However, following a technique similar to~\cite{webster2023transversal}, this is in fact possible. In fact, there are three related problems: (1) for fixed logical $H$ find transversal gates $P$, (2) for fixed $P$ find $H$, and (3) compute a generating set of all logical gates and their associated transversal implementations. All three of these problems can be posed as a system of linear equations over $\mathbb Z_8$, which as noted in \cite{webster2023transversal}, can be solved in polynomial time using the Howell normal form of a matrix over a ring. We sketch this procedure in Appendix~\ref{app:algos}.

\section{Conclusion}\label{sec:conclusion}

We have shown that spider-nest identities can be captured directly in terms of their associated triorthogonal matrices using the scalable ZX-calculus. Combining this fact with the graphical encoders for CSS codes introduced in \cite{kissinger2022grok} gives us a succinct and easy to digest characterisation for the set of transversal logical gates in $\mathcal D_3$ supported by a CSS code. The results we have shown here can be straightforwardly extended up the diagonal Clifford hierarchy, from $\mathcal D_3$ to $\mathcal D_\ell$. To do this, we can generalise the inductive family $s_n$ in Definition~\ref{def:sn} to $s_n^{(\ell)}$, whose base case $s_0^{(\ell)}$ has a $\pi/2^{\ell-1}$ phase. Then the analogue to the S4 equation is that $s_{\ell+1}^{(\ell)}$ separates. Interestingly, if we assume this equation for any fixed $\ell$, it automatically follows for any $\ell' \leq \ell$. The rest of the proof goes through by noting that Theorem~\ref{thm:indicator} generalises from triorthogonal to $\ell$-orthogonal matrices and degree $n - \ell - 1$ polynomials. Hence, we can prove all spider nest identities up the $\ell$-th level of the Clifford hierarchy just by adding one rule. However, it seems that proving spider nests at \textit{all} levels of the Clifford hierarchy still requires infinitely many rules. Note that this argument also holds for $\mathcal D_2$, so that we also characterise all the transversal diagonal Clifford unitaries. The proofs then don't require any additional spider-nest identity since the calculus is already complete for Cliffords.

Natural next steps are looking at non-diagonal transversal gates or non-CSS stabiliser codes. The story for diagonal gates in general stabiliser codes is relatively straightforward, given any stabiliser encoder can be decomposed into a CSS part and a diagonal part (cf. \cite{webster2023transversal}, Appendix C). There seems to be a nice graphical story there as well, relating to normal forms of Clifford ZX-diagrams, but we leave this, along with further explorations of the graphical structure of non-CSS codes, as future work. It also remains an open question whether Clifford+S4 (or Clifford+S$\ell$) rules can be extended naturally to a complete set of equations for the appropriate class of ZX-diagrams. While other complete axiomatisations exist, constructing the rules this way could provide new insights into working effectively with quantum computations at higher levels of the Clifford hierarchy.

\paragraph{Acknowledgements.} This work is supported by Engineering and Physical Sciences Research Council Grant EP/Y004736/1 ``Compilation and Verification of Quantum Software in the Noisy and Approximate Regime.''

\bibliographystyle{eptcs}
\bibliography{main}

\appendix

\section{Indicator polynomials for triorthogonal matrices}\label{app:indicator}

Since it is possible to represent arbitrary functions as polynomials, we can think of a polynomial $P$ in $n$ variables as a vector $[P]$ in $\mathbb F_2^{2^n}$ where $[P]_{\vec b} = P(\vec b)$. By restricting to polynomials of a fixed degree $r$, we obtain certain subspaces called Reed-Muller codes.

\begin{definition}
  The \textit{Reed-Muller code} $\RM(r, m)$ is the linear subspace of $\mathbb F_2^{2^m}$ consisting of vectors of the form $[P]$ for some polynomial $P \in \mathbb F_2[x_1, \ldots, x_m]$ of degree $\leq r$.
\end{definition}

A classic property of Reed-Muller codes is their orthocomplements are also Reed-Muller codes. This fact will help establish a correspondence with triorthogonal matrices.

\begin{theorem}\label{thm:rm-dual}
  For any $r < m$, $\RM(r, m)^\perp = \RM(m - r - 1, m)$.
\end{theorem}

\begin{proof}
  First, we note that if a polynomial $P$ of $m$ variables has degree $< m$, then $\sum_{\vec b} P(\vec b) = 0\ (\mod 2)$. This is easy to check for monomials, as any monomial of degree $< m$ must omit some variable $x_j$, hence
  \[ \sum_{\vec b} P(\vec b) = \sum_{\vec b, b_j = 0} P(\vec b) + \sum_{\vec b, b_j = 1} P(\vec b) = 0 \ (\text{mod } 2) \]
  The result holds for all polynomials by $\mathbb F_2$-linearity of the map $[P] \mapsto \sum_{\vec b} P(\vec b)\ (\text{mod } 2)$. Now, for any polynomial $P$ of degree at most $r$ and $Q$ of degree at most $m - r - 1$, $PQ$ has degree at most $m - 1$. Hence $[P] \cdot [Q] = \sum_{\vec b} PQ(\vec b) = 0\ (\text{mod }2)$. This implies $\RM(m - r - 1, m) \subseteq \RM(r, m)^\perp$. Since $\RM(r, m)$ has the monomials as its basis, $\dim(\RM(r, m)) = \sum_{d=0}^r\binom{m}{d}$. By manipulating binomial coefficients, we can see that:
  \[
  \dim(\RM(r,m)) + \dim(\RM(m-r-1, m))
  \]\[
  =
  \sum_{d=0}^r\binom{m}{d} + \sum_{d=0}^{m-r-1}\binom{m}{d} =
  \sum_{d=0}^r\binom{m}{d} + \sum_{d=r+1}^{m}\binom{m}{d} =
  2^m = \dim(\mathbb F_2^{2^m})
  \]
  so $\RM(m-r-1,m) = \RM(r, m)^\perp$.
\end{proof}

This enables us to show that semi-triorthogonal matrices are closely related to Reed-Muller codes. We give a proof here similar to the one given for unital triorthogonal spaces by Nezami and Haah~\cite{nezami2022classification}.

\begin{theorem}
  A matrix $M$ with $n$ columns  is semi-triorthogonal if and only if its indicator polynomial $P_M$ is of degree at most $n - 4$.
\end{theorem}

\begin{proof}
  Let $M'$ be a matrix obtained from $M$ by removing all repeated pairs of rows. $M$ is semi-triorthogonal if and only if $M'$ is, and both matrices have the same indicator polynomial. Hence, we can assume without loss of generality that $M$ has no repeated rows.

  Now, let $M$ have indicator polynomial $P = P_M$. Then $Q = x_j P$ is a polynomial with the property that $Q(\vec b) = 1$ if and only if $b_j = 1$ and $P(\vec b) = 1$, hence $\sum_{\vec b} Q(\vec b)$ is equal to the Hamming weight of the $j$-th column of $M$. This also works for products of columns: for $R = x_i x_j x_k P$, $\sum_{\vec b} R(\vec b)$ is equal to the Hamming weight of the element-wise product of the $i, j$ and $k$-th rows. Noting that $\sum_{\vec b} R(\vec b) \ (\text{mod } 2) = [x_i x_j x_k] \cdot [P]$, where $\cdot$ is the dot-product of vectors in $\mathbb F_2^{2^n}$, we see that $[P]$ must be orthogonal to all degree-3 monomials $[x_i x_j x_k]$. Since the latter span $\RM(3, n)$, $P \in RM(3, n)^\perp$, so by Theorem~\ref{thm:rm-dual} $P \in \RM(n - 4, n)$.
\end{proof}

\section{Computing transversal logical operations efficiently}\label{app:algos}

Following a technique similar to~\cite{webster2023transversal}, we will show that for fixed $(L, S)$ defining a CSS code, we can efficiently calculate matrices $H$ and $P$ making:
\[
    M = \left(\begin{array}{c|c}
      H & 0 \\ \hline
      PL & PS
    \end{array}\right)
\]
triorthogonal.

To see this, first note that we can span the space of all logical operators $D_H^\dagger$ using phase gadgets of degree at most 3, as larger phase gadgets can always be decomposed using spider nest identities~\cite{amy2017dihedral}. Hence, we can replace $H$ with $QK$, where $Q$ is a matrix whose rows all have Hamming weight 1, and $K$ has $O(k^3)$ rows consisting of all bitstrings of Hamming weight $\leq 3$. Thus $M$ becomes:
\[
    M = \left(\begin{array}{c|c}
      QK & 0 \\ \hline
      PL & PS
    \end{array}\right)
\]

Since the order of rows in $M$ is not relevant, the only function of $P$ and $Q$ is to fix the multiplicity of rows in the matrix:
\[
    N = \left(\begin{array}{c|c}
      K & 0 \\ \hline
      L & S
    \end{array}\right)
\]
Suppose $N$ is an $m \times n$ matrix. Form a new matrix $\widehat N$ over $\mathbb Z_8$ whose columns are labelled by the rows of $N$ and whose rows are labelled by all sets of at most 3 columns of $N$. Then let:
\[
  \widehat N_{S, i} = \begin{cases}
     2^{|S|-1} & \textrm{ if } \forall j \in S . N_{i,j} = 1 \\
     0 & \textrm{ otherwise}
  \end{cases}
\]

In other words, the $i$-th column of $\mathcal M$ corresponds to the inverse Fourier transform of the $j$-th row of $N$. Our goal is to find multiplicities for the rows of $N$ such that the sum of the inverse Fourier transform of each phase gadget adds to $0 \md 8$ on all sets of 1, 2, or 3 qubits. In other words, we should find a \textit{multiplicity vector} $\vec m \in \mathbb Z_8^m$ such that $\widehat N \vec m = \bm 0$.

Given a multiplicity vector $\vec m$, we can make $N$ into a triorthogonal matrix $M$ by repacing each row $i$ with $\vec m_i$ copies of that row. The set of all multiplicity vectors making $N$ into a triorthogonal matrix is therefore the kernel of $\widehat N$.

As noted in \cite{webster2023transversal}, we can find a generating set for the kernel of a $\mathbb Z_8$ matrix in polynomial time using its Howell normal form. Similarly, we can fix multiplicities of the upper or lower half of the matrix and obtain an associated inhomogeneous system of equations, which can also be solved in polynomial time.

\end{document}